\documentclass[journal,10pt,draftclsnofoot,onecolumn]{IEEEtran}
\usepackage[ruled,vlined,linesnumbered]{algorithm2e}
\usepackage{algpseudocode}
\usepackage{amssymb, amsmath, amsthm}
\usepackage{graphicx, subfigure}
\usepackage{url, xcolor}
\usepackage{verbatim, booktabs}
\usepackage{booktabs}
\usepackage{cite}

\DeclareMathOperator*{\argmin}{arg\,min}
\DeclareMathOperator*{\argmax}{arg\,max}

\newcommand{\ls}[1]  %% 1 in brackets means \ls takes 1 argument
   {\dimen0=\fontdimen6\the=#1\dimen0
    \advance\lineskip.5\fontdimen5\the\lineskip-\dimen0
    \lineskiplimit=.9\lineskip
    \baselineskip=\lineskip
    \advance\baselineskip\dimen0
    \normallineskip\lineskip
    \normallineskiplimit\lineskiplimit
    \normalbaselineskip\baselineskip
    \ignorespaces
   }

\begin{document}
\newtheorem{theorem}{Theorem}
\newtheorem{coro}{Corollary}[theorem]
\newtheorem{definition}{Definition}
\newtheorem{fact}{Fact}
\newtheorem{lemma}{Lemma}
\newtheorem{propo}{Proposition}
%\newtheorem{protocol}{Protocol}
%\newtheorem{remark}{Remark}
%\newtheorem{test}{Test}
%\newtheorem{transform}{Conversion Algorithm}

%\title{User Association in a Massive MIMO System with Small Cells}
%\title{User Association in Massive MIMO Heterogeneous Networks}
\title{User Association in Massive MIMO HetNets}

%\author{\IEEEauthorblockN{Yi Xu, Guosen Yue and Shiwen Mao}
%}

\author{Yi~Xu~\IEEEmembership{Student~Member,~IEEE}~and~Shiwen~Mao,~\IEEEmembership{Senior~Member,~IEEE}%
%\thanks{Copyright\copyright 2012 IEEE. Personal use of this material is permitted. However, permission to use this material for any other purposes must be obtained from the IEEE by sending a request to pubs-permissions@ieee.org.}
%\thanks{This work is supported in part by the National Science Foundation (NSF) under Grants CNS-0953513, CNS-1247955, and CNS-1320664.} 
\thanks{Y. Xu and S. Mao are with the Department of Electrical and Computer Engineering, Auburn University, Auburn, AL 36849-5201. Email: YZX0010@tigermail.auburn.edu, smao@ieee.org.}
\thanks{Shiwen Mao is the corresponding author: smao@ieee.org, Tel: (334)844-1845, Fax: (334)844-1809.}
%\thanks{Color versions of one or more of the figures in this paper are available online at http://ieeexplore.ieee.org.}
%\thanks{Digital Object Identifier XXXX/YYYYYY}
}

\maketitle

%%%\markboth{IEEE SYSTEMS JOURNAL (Under Review)}%
%{XU and MAO: USER ASSOCIATION IN A MASSIVE MIMO SYSTEM WITH SMALL CELLS}
%{XU and MAO: USER ASSOCIATION IN MASSIVE MIMO HETEROGENEOUS NETWORKS}
%%%{XU and MAO: USER ASSOCIATION IN MASSIVE MIMO HETNETS}

\maketitle

\ls{0.88}

\begin{abstract}
Massive MIMO and small cell are both recognized as the key technologies for the future 5G wireless systems.  In this paper, we investigate the problem of user association in a heterogeneous network (HetNet) with massive MIMO and small cells, where the macro base station (BS) is equipped with a massive MIMO and the picocell BS's are equipped with regular MIMOs. We first develop centralized user association algorithms with proven optimality, considering various objectives such as rate maximization, proportional fairness, and joint user association and resource allocation. We then model the massive MIMO HetNet as a repeated game, which leads to distributed user association algorithms with proven convergence to the Nash Equilibrium (NE). We demonstrate the efficacy of these optimal schemes by comparison with several greedy algorithms through simulations.
%We develop both centralized and distributed optimal algorithms. 
%Massive MIMO and small cell are both recognized as the key technologies for the future 5G wireless sytems. 
%techniques have drawn enormous attention recently.
%In this paper, we investigate the problem of user association in a heterogeneous network (HetNet) with massive MIMO and small cells, and develop both centralized and distributed optimal algorithms. 
%To investigate the user association problem is of great importance for system throughput enhancement.
%From the centralized and distributed perspective, we investigate how to obtain the optimal user association scheme. In particular, on one hand, we consider centralized user association for rate maximization, rate maximization with proportional fairness and $\log$ rate maximization with joint resource allocation design. On the other hand, we consider distributed user association when service provider sets the price or users bid for connection. For all the problems considered, we obtain the optimal solution. 
%We demonstrate the efficacy of these optimal schemes by comparing them to several greedy heuristic schemes through simulations.
\end{abstract}


\pagestyle{empty}\thispagestyle{empty}
%\pagestyle{plain}\thispagestyle{plain}
%\pagestyle{headings}\thispagestyle{headings}

%\begin{keywords}
%Interference Alignment.
%\end{keywords}

%-----------------------------------------------------------
\section{Introduction} \label{sec:intro}
%-----------------------------------------------------------

Over the past two decades, Multiple Input Multiple Output (MIMO) has evolved from a pure theory to a practical technology, and has greatly enhanced the wireless system capacity by offering many degrees of freedom (DoF) for wireless transmissions.
However, due to the so-called ``smartphone'' revolution, mobile users are demanding increasingly higher data rates for rich multimedia applications. Existing and future wireless networks are facing the grand challenge of a 1000-time increase in mobile data in the near future~\cite{Qualcomm}. 
There have been tremendous efforts made aiming to cater for this demand. For example, based on MIMO and OFDM, LTE-Advanced targets at a peak rate of $1$ Gbps, but the average rate is still less than $100$ Mbps. 
In the foreseeable future, 
%with more and more video related data traffic \cite{cisco2013}, 
such rates can hardly be satisfactory for data-hungry wireless users. 

To boost wireless capacity, two technologies have gained most attention from both industry and academia. 
The first one is massive MIMO (a.k.a., large-scale MIMO, full-dimension MIMO, or hyper MIMO)~\cite{Larsson13,Rusek13}. The idea 
%of massive MIMO 
is to equip a base station (BS) with hundreds, thousands, or even tens of thousands of antennas, hereby providing an unprecedented level of DoF for mobile users. The massive MIMO concept has been successfully demonstrated in recent works~\cite{LZhong12,LZhong13}.
The second technology is small cell. 
A great benefit of deploying small cells is that the distance of the user-BS link can be effectively reduced, leading to reduced transmit power, higher data rate, enhanced coverage, and better spatial reuse of spectrum. Both massive MIMO and small cells are recognized as key technologies of the future $5$G wireless systems~\cite{andrew5g}. 

In this paper, we consider a heterogeneous network (HetNet) with massive MIMO and small cells, 
%in particular, 
where the macrocell BS (MBS) is equipped with a massive MIMO and the picocell BS's (PBS) are equipped with regular MIMOs.
%
%The trend of merging massive MIMO and small cell technologies to form HetNet (heterogeneous networks) has become more and more bright and clear. Many researchers and standardization organizations have considered them as the core technologies for $5$G. \cite{andrew5g} views massive MIMO and small cell as two of the ``big three'' technologies for $5$G wireless communication system.
%
%Given the envisioned benefits of massive MIMO and small cells, to combine these two technologies, the first question would be how to associate the users and the BSs, so that the system throughput or user experience can be ultimately enhanced.
To fully harvest the benefits promised by these two technologies in an integrated HetNet system, it is critical to investigate the user association problem, i.e., how to assign active users to the BS's such that the system-wide capacity can be maximized and users' experience can be enhanced.  

There are already several recent works pushing forward in this direction.
In~\cite{Giuseppe13,Giuseppe132,yixuicc,yiaccess}, the authors consider the problem of user association in massive MIMO systems operated in the frequency-division duplexing (FDD) mode. These papers are focused on a macrocell without small cells. In~\cite{Bethanabhotla}, user association in time-division duplexing (TDD) massive MIMO system is addressed, where factional user association is allowed. 
%However, it is worth pointing out that fractional user association is allowed in \cite{Bethanabhotla}.
Bayat et al. in~\cite{Bayat2014} model the problem of user association in a femtocell HetNet as a dynamic matching game and derive the optimal user association. However, massive MIMO is not considered in the system model.
In~\cite{Gupta14}, the authors investigate the problem of user association with conventional MIMO BS's and propose a simple bias based selection criterion to approximate more complex selection rules.
Bj$\ddot{o}$rnson, et al. in~\cite{Debbahmmsm} consider the problem of improving the energy efficiency without sacrificing the quality of service (QoS) of users in a massive MIMO and small cell HetNet.

%Different from these works, 
Motivated by these interesting works, we consider the user association problem in a TDD massive MIMO HetNet in this paper, taking into consideration of the practical constraints, such as the limited load capacity at each BS, while without allowing fractional user association. 
The main goal is to maximize the system capacity while enhancing user experience. 

More specifically, this paper contains two parts: (i) centralized user association and (ii) distributed user association.
For centralized user association, we investigate the problems of rate maximization, rate maximization with proportional fairness, and joint resource allocation and user association. We prove the unimodularity of our formulated problem and 
%leverage the unimodularity to obtain 
develop optimal user association algorithms to the problems of rate maximization and rate maximization with proportional fairness. We then propose a series of primal decomposition and dual decomposition algorithms to solve the problem of joint resource allocation and user association and prove the optimality of the proposed scheme.
For distributed user association, we model the behavior and interaction between the service provider, who owns the BS's, and users as repeated games. We consider two types of operations: (i) the service provider sets the price and the users decide which BS to connect to, and (ii) the users bid for the opportunity of connection. We prove that in both cases the the proposed algorithms converge to the respective Nash Equilibrium (NE). 
%repeated games converge. The proposed algorithms to reach the Nash Equilibrium (NE) of each game.

%The reminder of this paper is organized as follows.
In the reminder of this paper, Section~\ref{sec:sysmo} introduces the system model and preliminaries.
Optimal centralized and distributed user association schemes are presented in Sections~\ref{sec:central} and~\ref{sec:dist}, respectively.
Section~\ref{sec:sim} presents the simulation study and Section~\ref{sec:con} concludes this paper. 
%Conclusion is drawn in Section \ref{sec:con}.
Throughout this paper, we use a boldface upper (lower) case symbol to denote a matrix (vector), and a normal symbol to denote a scalar. $(\cdot)^H$ denotes the Hermitian of a matrix.

%----------------------------------------------
\begin{comment}

\begin{table} [!t]
\begin{center}
\caption{Notation}
\label{tab:nota}
\begin{tabular}{l|l}
\toprule
\textbf{Symbol} & \textbf{Definition}\\
\midrule
$\mathbf{V}$ & Precoding matrix \\
$x_{k_g}$ & Connection indicator \\
$\mathbf{H}$ & Channel matrix \\
$\mathbf{d}$ & Data vector \\
$\mathbf{n}$ & Noise vector \\
$R_{k_j}$ & Rate of user $k$ connecting to BS $j$\\
$L_j$ & Size of user to serve by BS $j$ \\
$M_j$ & Antenna number of BS $j$ \\
$P_j$ & Transsmission power of BS $j$ \\
$d_{j,k}$ & Distance between user $k$ and BS $j$ \\
$\mathbf{A}$ & Constraint matrix \\
$\mathbf{S}_i$ & Square submatrix of size $i$ of $\mathbf{A}$\\
$\zeta_i$ & Number of $1$s in the $i$-th column of $\mathbf{A}$ \\
$\varphi$ & Power normalization factor\\
$\eta_k$ & Sum rate for user $k$ \\
$\eta_{k_j}$ & Achievable rate for user $k$ regarding BS $j$\\
$\mathbf{x}$ & Vector for optimization variable\\
$\mathbf{c}$ & Vector for optimization objective\\
$\mathbf{A}$ & Constraint matrix \\
$\mathbf{V}$ & Users' evaluation matrix \\
$\mathcal{G}_j$ & User set with coverage of BS $j$\\
\bottomrule
\end{tabular}
\end{center}
%\vspace{-0.15in}
\end{table}

\end{comment}
%----------------------------------------------------------

%-----------------------------------------------------------
\section{System Model and Preliminaries} \label{sec:sysmo}
%-----------------------------------------------------------

%%Challenges:
%%\begin{itemize}
	%%\item how do users know rate
	%%\item BS how to know which users
	%%\item Massive MIMO coupling
	%%\item Load of MBS and PBS $\leq$ number of antenna
	%%\item Scheduling problem: if a user wants to connect, however, the BS does not schedule it.
	%%\item How to make decisions based on local information
%%\end{itemize}
%%
%%Questions:
%%\begin{itemize}
	%%\item Relationship between the number of antenna and number of user?
%%\end{itemize}

%\begin{figure} [!t] %[thb]
%\center{\includegraphics[width=3.2in, height=2.0in]{ scenario.eps}}
%\caption{Illustration of a Massive MIMO system with small cells.}
%\label{fig:sce}
%%\vspace{-0.15in}
%\end{figure}

The system considered in this paper includes $K$ users and $J$ BS's, including an MBS
%a macrocell BS (MBS) 
with a massive MIMO and $(J-1)$ 
%picocell BS's (PBS), 
PBS's, each equipped with a conventional MIMO. 
%as illustrated in Fig.~\ref{fig:sce}. 
%As stated in~\cite{Damnjanovic}, picocells can benefit from inter-cell interference coordination (ICIC).
The channel model is
%\begin{eqnarray}
	$h_{j,k,n}=g_{j,k,n} l_{j,k}$,
%\label{eqn:mmhn99}
%\end{eqnarray}
where $h_{j,k,n}$ is the channel of antenna $n$ at BS $j$ to user $k$, $g_{j,k,n}$ represents the small scale fading coefficient between antenna $n$ of BS $j$ and user $k$, and $l_{j,k}$ stands for the large scale fading coefficient between BS $j$ and user $k$~\cite{lulu2014}. Concatenating all the channel coefficients from all the antennas of BS $j$, we obtain the channel vector $\mathbf{h}_{j,k}$, 
%Thus $\mathbf{h}_{j,k}$ is the channel vector from the $j$-th BS to user $k$.
%Putting the channels from all users along the column, we denote $\mathbf{H}_j = \left[ \mathbf{h}_{j,1}, \mathbf{h}_{j,2}, \cdots, \mathbf{h}_{j,k} \right]$ as the channel coefficient matrix for signals transmitted from the $j$-th BS. 
as well as the channel coefficient matrix for signals transmitted from BS $j$ as $\mathbf{H}_j = \left[ \mathbf{h}_{j,1}, \mathbf{h}_{j,2}, \cdots, \mathbf{h}_{j,k} \right]$ . 

Let $\mathbf{y}_j$ denote the signals received by the users connecting to BS $j$, $\mathbf{W}_j$ the precoding matrix of BS $j$, and $\mathbf{d}_j$ the data sent from BS $j$. We have
\begin{eqnarray}
	\mathbf{y}_j = \mathbf{H}_j\mathbf{W}_j\mathbf{d}_j + \mathbf{n}_j,
\label{eqn:mmhn89}
\end{eqnarray}
where $\mathbf{n}_j$ is the zero mean circulant symmetric complex Gaussian noise vector. 

Each active user has the options to connect to either the MBS or a PBS. For a user $k$, define user association index variable $x_{k_j}$ as
\begin{equation}
  x_{k_j}=\begin{cases}
    1, & \text{if user $k$ is connected to BS $j$}.\\
    0, & \text{otherwise}.  
  \end{cases}
\label{eqn:mmhn95}
\end{equation}
%who may or may not connect to the massive MIMO base station, denoting  
Let its achievable rate if connected to BS $j$ be $R_{k_j}$, $\eta_{k_j}=x_{k_j} R_{k_j}$, 
%with respect to BS $j$ as $\eta_{k_j}$, 
and its actual data rate be $\eta_k$. %we readily have:
%\begin{eqnarray}
%	\eta_{k_j} = x_{k_j} R_{k_j},
%\label{eqn:mmhn98n}
%\end{eqnarray}
%where
%Denote $\eta_k$ as the sum rate for user $k$, then we have:
We have
\begin{eqnarray}  \label{eqn:mmhn98nn}
	\eta_k = \sum_j \eta_{k_j} = \sum_j x_{k_j} R_{k_j}.
\end{eqnarray}

For users connecting to a massive MIMO BS $j$ (i.e., the MBS), their achievable rate can be approximated with the following deterministic rate~\cite{Bethanabhotla}.
\begin{eqnarray}
	R_{k_j}= \log \left( 1 + \frac{M_j-L_j+1}{L_j}\frac{P_j  l_{j,k}}{1+\sum_{j' \neq j}{P_{j'}
	         l_{j',k}}} \right),
\label{eqn:mmhn98}
\end{eqnarray}
where $M_j$ is the number of antennas at the BS, $L_j$ is the prefixed load parameter of the BS indicating how many users it could serve, and $P_j$ is transmit power from the MBS. Note that there is no small scale fading factor in~(\ref{eqn:mmhn98}). This approximation has been proven to be accurate~\cite{Bethanabhotla}.

%As stated in~\cite{Damnjanovic}, picocells can benefit from inter-cell interference coordination (ICIC).
For a PBS with a conventional MIMO, we assume that the inter-cell interference is negligible  among the picocells, due to the small transmission powers and effective inter-cell interference coordination (ICIC)~\cite{Damnjanovic}. 
%of these small cell BSs are typically low. 
The achievable rate of user $k$ connecting to PBS $j$ can be represented as follows.
%\textbf{Conventional MU-MIMO}:
\begin{eqnarray}
	\widetilde{R}_{k_j} = \log \left(1+ \frac{P_j \left|\mathbf{h}^H_{j,k}\mathbf{w}_{j,k}\right|^2}{1+\sum_{k' \neq k}P_j\left|\mathbf{h}^H_{j,k}\mathbf{w}_{j,k'}\right|^2} \right),
\label{eqn:mmhn97}
\end{eqnarray}
where $\mathbf{w}_{j,k}$ is the $k$-th column of BS $j$'s precoding matrix $\mathbf{W}_j$. 
There are many precoding designs for conventional MIMO BS's, such as
%\begin{align}
%	\left\{ \begin{array}{l}
%	  \mathbf{W}_j = \frac{1}{\sqrt{\varphi}} \mathbf{H}_j^H \\ %\label{eqn:mmhn97w1} \\
%	  \mathbf{W}_j = \frac{1}{\sqrt{\varphi}} \mathbf{H}_j^H(\mathbf{H}_j^T\mathbf{H}_j^H)^{-1} \\ %\label{eqn:mmhn97w2} \\
%	  \mathbf{W}_j = \frac{1}{\sqrt{\varphi}} \mathbf{H}_j^H(\mathbf{H}_j^T\mathbf{H}_j^H + \delta \mathbf{I})^{-1}. %\label{eqn:mmhn97w3}
%	        \end{array} \right. 
%\end{align}
%The above equations are the precoding matrices for
matched filter (MF) precoding, zero forcing (ZF) precoding, and regularized zero forcing (RZF) precoding~\cite{lulu2014}. 
%, respectively. Note that here $\varphi$ is a power normalization factor.
%(\ref{eqn:mmhn97w1}) is Matched Filter(MF), (\ref{eqn:mmhn97w2}) is ZF, and (\ref{eqn:mmhn97w3}) is RZF.
%
Without loss of generality, we adopt MF precoding in this paper with $\mathbf{W}_j = \frac{1}{\sqrt{\varphi}} \mathbf{H}_j^H$, where $\varphi$ is a power normalization factor. The signal received by all the users connecting to PBS $j$ can be rewritten as follows.
\begin{equation}
\label{eqn:mmhn88}
\mathbf{y}_j = 
\begin{pmatrix}
h^H_{j,1}h_{j,1}d_1+h^H_{j,1}h_{j,2}d_2+\cdots+h^H_{j,1}h_{j,k}d_k \\
h^H_{j,2}h_{j,1}d_1+h^H_{j,2}h_{j,2}d_2+\cdots+h^H_{j,2}h_{j,k}d_k \\
\cdots \\
h^H_{j,k}h_{j,1}d_1+h^H_{j,k}h_{j,2}d_2+\cdots+h^H_{j,k}h_{j,k}d_k
\end{pmatrix}.
\end{equation}
Thus, the achievable rate for user $k$ regarding to PBS $j$ can be obtained as follows.
\begin{eqnarray}
\eta_{k_j} = \log\left(1 + \frac{P_j \left| x_{k_j} h^H_{j,k}h_{j,k}\right|^2}{1+\sum_{k' \neq k}P_j \left| x_{k'_j}h^H_{j,k}h_{j,k'}\right|^2}\right).
\label{eqn:mmhn87}
\end{eqnarray}

%-----------------------------------------------------------
\section{Centralized User Association} \label{sec:central}
%-----------------------------------------------------------
%%%\reminder{prove it's NP-hard}
%%%
%%%\reminder{Hungarian algorithm could solve when $k=j$ and $L_j =1$}
%%%
%%%\reminder{The constraint matrix is sparse. How to make use of this?}
%%%
%%%%s.t.: $f(\mathbf{V}_M,\mathbf{V}_P)$
%%%
%%%Fix the size of user to be served?
%%%
%%%
%%%How to approximate user SINR?
%%%\begin{itemize}
	%%%\item For MBS, user SINR is irrelevant to small scale fading, just related to large scale fading.
	%%%\item For PBS, user SINR is related to both large and small scale fading.
%%%\end{itemize}

%%%\textbf{Utility function}: (provide detailed discussion about these criteria.)
In this section, we consider the problem of centralized user association. We assume that the BS's have 
%obtained 
all the channel state information (CSI) via uplink training. We adopt the following utility function for each user $k$ with achievable rate $\eta_k$. 
%based on its achievable rate as
\begin{equation}
  \mathcal{U}(\eta_k)=\begin{cases}
	\eta_k^{1-\alpha} / (1-\alpha), & \text{if } \alpha>0,\alpha \neq 1 \\
	\eta_k, & \text{if } \alpha = 0 \\
    \log(\eta_k), & \text{if } \alpha = 1.
  \end{cases}
\label{eqn:mmhn96}
\end{equation}
%The implications behind~(\ref{eqn:mmhn96}) are as follows.
%\begin{itemize}
%	\item $\alpha = 0$ yields the maximization of the sum rate (no fairness).
%	\item $\alpha \rightarrow \infty$ yields the maximization of the worst-case rate (max-min fairness).
%	\item $\alpha = 1$ yields the maximization of the geometric mean rate (proportional fairness).
%\end{itemize}
When $\alpha = 0$, maximizing $\mathcal{U}(\cdot)$ yields the maximization of the sum rate (but no fairness); when $\alpha \rightarrow \infty$, it leads to the maximization of the worst-case rate (i.e, max-min fairness); when $\alpha = 1$, it yields the maximization of the geometric mean rate (i.e., proportional fairness).

Our goal is to maximize the system utility by configuring the user-BS association. Typically, we consider the cases when $\alpha = 0$ and $\alpha = 1$. In the case of $\alpha=1$, we define $\mathcal{U}(0)=0$.
%$\mathcal{U}(\eta_k)=0$, if $\eta_k = 0$.

%-----------------------------------------------------------
\subsection{Maximizing Sum-rate} \label{subsec:cms}
%-----------------------------------------------------------

We firstly investigate the problem of maximizing the system sum rate, i.e., $\alpha = 0$ in~(\ref{eqn:mmhn96}) and $\mathcal{U}(\eta_k)= \eta_k$. 
The problem 
%is to maximize $\sum_k \mathcal{U}_k(z)$ given the system configuration and user distribution, which is mathematically 
can be formulated as follows.
\begin{align}
\mbox{\textbf{P1-1:}} &\;\; \max_{\{x_{k_j}\}} \sum_{k=1}^K{\eta_{k}}  \label{eqn:mmhn94} \\ 
\mbox{s.t.} &\;\; \sum_k x_{k_j} \leq L_j \leq M_j,~j=1,2,\cdots,J  \nonumber \\ 
            &\;\; \sum_j x_{k_j} \leq 1,~k=1,2,\cdots,K \nonumber \\ 
            &\;\; \mbox{Constraints~(\ref{eqn:mmhn95}),~(\ref{eqn:mmhn98nn}),
            (\ref{eqn:mmhn98}),~(\ref{eqn:mmhn87})}. \nonumber
%            &\;\; \mbox{Constraints (\ref{eqn:mmhn98}),%~(\ref{eqn:mmhn98n}),
%            ~(\ref{eqn:mmhn95}),~(\ref{eqn:mmhn98nn}),~(\ref{eqn:mmhn87})}. \nonumber 
\end{align}

Note that the second constraint requires the number of users connecting to a BS to be no more than its prefixed load, which should in turn be no more than the number of antennas it has, since theoretically BS $j$ can provide at most $M_j$ degrees of freedom (DoF). 
%But constraint $L_j \leq M_j$ should be satisfied by when configuring the system parameters. So we could drop this constraint out of the optimization problem.
Assuming the $L_j$'s are already chosen to satisfy $L_j \leq M_j$, we drop this constraint in the remainder of this paper. The third constraint simply claims that each user can connect to at most one BS at a time.

A key observation is that~(\ref{eqn:mmhn87}) can be rewritten as
\begin{eqnarray}
	\eta_{k_j} = x_{k_j} \log\left(1 + \frac{P_j \left|  h^H_{j,k}h_{j,k}\right|^2}{1+\sum_{k' \neq k}P_j 	\left| x_{k'_j}h^H_{j,k}h_{j,k'}\right|^2}\right). 
\label{eqn:mmhn87n}
\end{eqnarray}
Thus the $\widetilde{R}_{k_j}$ in (\ref{eqn:mmhn97}) can be redefined as
\begin{eqnarray}
	\widetilde{R}_{k_j} = \log\left(1 + \frac{P_j \left|  h^H_{j,k}h_{j,k}\right|^2}{1+\sum_{k' \neq k
	}P_j \left| x_{k'_j}h^H_{j,k}h_{j,k'}\right|^2}\right).
\label{eqn:mmhn79o}
\end{eqnarray}

In~(\ref{eqn:mmhn79o}), it can be seen that $\widetilde{R}_{k_j}$ depends on other users' choices $x_{k'_j}$, for all $k \neq k'$, as well. To make the problem tractable, we adopt the worst-case approximation by assuming the users within the coverage of BS $j$ (denoted as $\mathcal{G}_j$) all connect to BS $j$ with perfect channels. This way,~(\ref{eqn:mmhn79o}) can be approximated as
\begin{eqnarray}
	\widetilde{R}_{k_j} = \log\left(1 + \frac{P_j \left|  h^H_{j,k}h_{j,k}\right|^2}{1+ (\left| \mathcal{G}_j \right| -1)P_j }\right), 
\label{eqn:mmhn79}
\end{eqnarray}
where $\left| \cdot \right|$ for a set stands for the cardinality of the set.

Define auxiliary variables $c_{k_j}$ as follows.
\begin{eqnarray}
  c_{k_j}=\begin{cases}
    R_{k_j} \text{ in (\ref{eqn:mmhn98})}, & \text{if BS $j$ is the MBS};\\
    \widetilde{R}_{k_j} \text{ in (\ref{eqn:mmhn79})}, & \text{if BS $j$ is a PBS} .
  \end{cases}
\label{eqn:mmhn80}
\end{eqnarray}
The sum rate maximization problem can be reformulated as
\begin{align}
\mbox{\textbf{P1-2:}} &\;\; \max_{\{x_{k_j}\}} \sum_{k=1}^K\sum_{j=1}^J {x_{k_j}c_{k_j}} \label{eqn:mmhn78} \\
\mbox{s.t.} &\;\; \sum_k x_{k_j} \leq L_j,~j=1,2,\cdots,J \nonumber \\
            &\;\; \sum_j x_{k_j} \leq 1,~k=1,2,\cdots,K \nonumber \\
            &\;\; \mbox{Constraints (\ref{eqn:mmhn95}),~(\ref{eqn:mmhn80})}. \nonumber
\end{align}

%For the above optimization problem, since its 
Since the variables $x_{k_j}$'s are binary, problem {\bf P1-2}  falls into the category of \textit{Multiple Knapsack Problems}, which is one of Karp's $21$ NP-complete problems~\cite{Karp72}. 
%At this point, one may try to use a greedy algorithm to obtain sub-optimal solutions. However, by taking advantage of the coefficients of the constraints, we could obtain the optimal solution.
Although a greedy algorithm could be developed to compute sub-optimal solutions, we show that problem {\bf P1-2} can actually be optimally solved by taking advantage of its special structure. 

Let $\mathbf{X}$ be a matrix with entries $x_{k_j}$, $k = 1, 2, \cdots, K$, $j = 1, 2, \cdots, J$. We could convert $\mathbf{X}$ to a vector $\mathbf{x}$ by concatenating the rows of $\mathbf{X}$ and taking a transpose 
%For instance, the following matrix 
%\begin{eqnarray}
%\label{eqn:mmhn77}
%\mathbf{X} &=& 
%	\begin{pmatrix}
%		x_{1_1}& x_{2_1} & \cdots &x_{K_1} \\
%		x_{1_2}& x_{2_2} & \cdots &x_{K_2} \\
%		\cdots & \cdots  & \cdots &\cdots  \\
%		x_{1_J}& x_{2_J} & \cdots &x_{K_J}
%	\end{pmatrix}
%	.
%\end{eqnarray}
%is reshaped 
as $\mathbf{x}=\left[x_{1_1}~x_{2_1}~\cdots~x_{K_1}~\cdots~x_{1_J}~\cdots~x_{K_J} \right]^T$, and 
simplify the notation as $\mathbf{x}=\left[x_1~x_2~\cdots~x_{K_J} \right]^T$. We then apply the same conversion to the matrix comprising $c_{k_j}$ and obtain vector $\mathbf{c}$. 
%With these transformations, we rewrite the optimization problem as:
Problem \textbf{P1-2} can be rewritten as
\begin{align}
\mbox{\textbf{P1-3:}} &\;\; \max_{\mathbf{x}} \mathbf{c}^T \mathbf{x} \label{eqn:mmhn76} \\
\mbox{s.t.} &\;\; \sum_{k=1}^K x_{(j-1)K+k} \leq L_j,~j=1,2,\cdots,J \nonumber \\
            &\;\; \sum_{j=1}^J x_{k+(j-1)K} \leq 1,~k=1,2,\cdots,K \nonumber \\
            &\;\; \mbox{Constraints (\ref{eqn:mmhn95}),~(\ref{eqn:mmhn80})}. \nonumber
\end{align}

Ignoring constraints~(\ref{eqn:mmhn95}) and~(\ref{eqn:mmhn80}), define $\mathbf{A}$ as the constraint matrix of problem \textbf{P1-3}, with entries being the coefficients of the first and second constraints. 
We next introduce an important definition and derive a key lemma.

\begin{definition} 
A matrix $\mathbf{A}$ is called totally unimodular if the determinant of every square submatrix of $\mathbf{A}$ is either $0$, $+1$ or $-1$~\cite{Schrijver}.
\end{definition}

\begin{lemma}
The constraint matrix $\mathbf{A}$ of problem \textbf{P1-3} is totally unimodular.
\label{lem1}
\end{lemma}
\begin{proof}
%To prove that $\mathbf{A}$ is totally unimodular, we need to check if every square submatrix of $\mathbf{A}$ has a determinant of either $0$, $+1$, or $-1$.
Inspecting the constraints in problem \textbf{P1-3}, we find that $\mathbf{A}$ is of the following form.
\begin{eqnarray}
\label{eqn:mmhn81}
\mathbf{A} &=& 
\begin{pmatrix}
1~1 \cdots 1 & 0~0 \cdots 0 & 0~0 \cdots 0 \\
0~0 \cdots 0 & 1~1 \cdots 1 & 0~0 \cdots 0 \\
\vdots & \vdots & \vdots \\
0~0 \cdots 0 & 0~0 \cdots 0 & 1~1 \cdots 1 \\

1~0 \cdots 0 & 1~0 \cdots 0 & 1~0 \cdots 0 \\
0~1 \cdots 0 & 0~1 \cdots 0 & 0~1 \cdots 0 \\
\ddots & \ddots & \ddots \\
0~0 \cdots 1 & 0~0 \cdots 1 & 0~0 \cdots 1 \\
\end{pmatrix}.
\end{eqnarray}
We can divide $\mathbf{A}$ into blocks as follows. 
\begin{eqnarray}
\label{eqn:mmhn81n}
\mathbf{A} &=& 
\begin{pmatrix}
\mathbf{A}_1& \mathbf{A}_2 & \cdots &\mathbf{A}_J \\
\mathbf{B}_1& \mathbf{B}_2 & \cdots &\mathbf{B}_J
\end{pmatrix}
,
\end{eqnarray}
where each $\mathbf{A}_j,~j\in [1,J]$, is a submatrix of $\mathbf{A}$ of size $J \times K$; and each $\mathbf{B}_j,~j \in [1,J]$, is an identity matrix of size $K \times K$.

Let $\mathbf{S}_{n}$ denote an arbitrary square submatrix of matrix $\mathbf{A}$ of size $n$.
For any submatrix of $\mathbf{A}$ of size $n=1$, it is trivial to see that the determinant of this submatrix is either $0$ or $+1$.
%, since all of the entries of the constraint matrix are either $0$ or $+1$. 
So we only need to consider the case where the size of the square submatrix is greater than or equal to $2$. 
%i.e., $\mathbf{S}_{n}$ with $n \geq 2$.

Case 1: $\mathbf{S}_{n}$ is taken entirely from one of the submatrices $\mathbf{A}_j$ or $\mathbf{B}_j$, $j \in [1,~J]$.
We can see from the structure 
%of $\mathbf{A}_j$ 
that at least one row of $\mathbf{A}_j$ is all zero. So if the square submatrix is entirely taken from $\mathbf{A}_j$, the determinant of the submatrix is zero. Since matrix $\mathbf{B}_j$, for all $j$, is simply an identity matrix, it is straightforward that the determinant of any square submatrix of $\mathbf{B}_j$ is either $0$ or $+1$.

Case 2: $\mathbf{S}_{n}$ is not entirely taken from any one of the submatrices $\mathbf{A}_j$ or $\mathbf{B}_j$, $j \in [1, J]$.
In this case, the square submatrix must be taken from $2n$ ($n = 1, \cdots, J$) submatrices of the submatrix set $(\mathbf{A}_j \cup \mathbf{B}_j$, $j \in {1,\cdots, J})$. We next proceed with our proof by applying induction method.

For the base case $n=1$, the square submatrix to be examined is of size $2$. Since the entries can only be $0$ or $+1$, the determinant 
%of the square submatrix 
can only be $0$, $+1$ or $-1$.

Now assuming that any square submatrix of size $(n-1)$ has determinant $0$, $+1$ or $-1$, we need to check if the same conclusion holds for any square submatrix of size $n$.

We first notice that each column of $\mathbf{A}$ has exactly two $+1$s. Moreover, exactly one of them is in $\mathbf{A}_j$, and the other in $\mathbf{B}_j$. 
Let $q^* = \argmin_q \sum_{i} \mathbf{S}_{n_{i,q}}$, where $\mathbf{S}_{n_{i,q}}$ is the $(i,q)$-th entry of $\mathbf{S}_{n}$.
That is, column $q^*$ has the minimum number of $1$s among all the columns of $\mathbf{S}_n$.

Let $\zeta_{q^*} = \min_q \sum_{i} \mathbf{S}_{n_{i,q}}$. $\zeta_{q^*}$ can only be $0$, $1$, or $2$.

If $\zeta_{q^*} = 0$, then all the entries of the $q^*$-th column of $\mathbf{S}_n$ are $0$, which results in $det(\mathbf{S}_n) = 0$, where $det$ is short for determinant.

If $\zeta_{q^*} = 1$, then we could calculate $\mbox{det}(\mathbf{S}_n)$ by expanding the $q^*$-th column and obtain $\mbox{det}(\mathbf{S}_n) = \mbox{det}(\mathbf{S}_{(n-1)})$. Since $\mbox{det}(\mathbf{S}_{(n-1)})$ is $0$, $1$ or $-1$ by our induction hypothesis, we conclude $\mbox{det}(\mathbf{S}_n)$ is $0$, $1$ or $-1$.

If $\zeta_{q^*} = 2$, we could firstly negate all the entries taken from $\mathbf{B}_j$, and then add all the rows in $\mathbf{B}_j$ to any non-zero row in $\mathbf{A}_j$. After this procedure, if that non-zero row in $\mathbf{A}_j$ is still non-zero, add that row to any other non-zero row in $\mathbf{A}_j$. Repeat this process until we get a zero row in $\mathbf{A}_j$. The reason why this process always give us a all-zero row is that we have equal number of $+1$s in $\mathbf{A}_j$ and $\mathbf{B}_j$. Since any basic row operation does not change the determinant and we finally get a all-zero row, we have $\mbox{det}(\mathbf{S}_n) = 0$. That completes our induction. 
%steps.
%To sum up, the determinant of any square submatrix is either $0$, $+1$ or $-1$.
\end{proof}

\begin{fact}
For a linear programming problem, if its constraint matrix satisfies totally unimodularity, then its has all integral vertex solutions~\cite{Schrijver}.
\label{lem2}
\end{fact}
%Proof of this lemma is skipped here. Interesting readers are referred to \cite{Schrijver}.

\begin{fact}
For a linear programming problem, if it has feasible optimal solutions, then at least one of them occurs at a vertex of the polyhedron define by its constraints~\cite{Berenstein}.
\label{lem3}
\end{fact}
%\begin{proof}
%Linear programming inequality constraints are half-spaces and equality constraints are hyperplanes. So all these constraints together define a polyhedron. By the maximum principle of convex functions \cite{Berenstein}, the optimal value can be obtained on the boundaries. Since the vertex is the intersection of several constraints, the optimal solution can be found by checking the vertices.
%\end{proof}

Given the facts and Lemma~\ref{lem1}, we have the following theorem. The proof is straightforward and omitted. 

\begin{theorem}
The optimal solution of problem~\textbf{P1} can be obtained by solving a relaxed problem where the variables $x_{k_j}$ are allowed to take real values between $[0,1]$.
\end{theorem}

%---------------------------------------------------
\begin{comment}

\begin{proof}
%
\begin{align}
\mbox{\textbf{NP1:}} &\;\; \max_{\mathbf{x}} \mathbf{c}^T \mathbf{x}  \label{eqn:mmhn75} \\
\mbox{s.t.} &\;\; \sum_{k=1}^K x_{(j-1)K+k} \leq L_j,~j=1,2,\cdots,J \nonumber \\
            &\;\; \sum_{j=1}^J x_{k+(j-1)K} \leq 1,~k=1,2,\cdots,K \nonumber \\
            &\;\; 0 \leq x_{k} \leq 1,~k=1,2,\cdots,K \nonumber \\
            &\;\; \mbox{Constraints (\ref{eqn:mmhn80})}. \nonumber
\end{align}

Problem~\textbf{NP1} is a relaxed of version \textbf{P1}. We can see that \textbf{NP1} is a linear programming problem. By Lemma \ref{lem3}, we know that at least one optimal solution of \textbf{NP1} is the vertex solution. Applying Lemmas~\ref{lem1} and~\ref{lem2} to \textbf{NP1}, we know that the vertex solutions are integers. Since the variable $x_k$ is restricted to the range of $[0,1]$, that means the optimal solutions of \textbf{NP1} are binary. Since the optimal value of \textbf{P1} is upper bounded by the optimal value of \textbf{NP1}, we could attain the optimal value of \textbf{P1} by setting variables according to the solution of \textbf{NP1}. That means \textbf{NP1} and \textbf{P1} are equivalent. We could solve \textbf{NP1} for \textbf{P1}.
\end{proof}

\end{comment}
%---------------------------------------------------

Given the above theorem, we could obtain the optimal solution of \textbf{P1} by solving the relaxed problem, termed \textbf{NP1}, using common LP solvers~\cite{Schrijver}.

%%%\reminder{What algorithm? Interior point method? Solution may not be at the boundary. Dual simplex? Worst case complexity is Exponential.}

%-----------------------------------------------------------
\subsection{Proportional Fairness} \label{subsec:cpf}
%-----------------------------------------------------------

%BS set: $\mathcal{J}=\left\{ 0, 1, 2, \cdots, J \right\}$, where $BS_0$ is the MBS, and $BS_j,j \neq 0$ is the PBS.
%
%User set: $\mathcal{K}=\left\{ 1, 2, \cdots, K\right\}$.

%%%\reminder{What if $\sum_j x_{k_j} = 0$? That is a user does not connect to any BS. Then utility is $-\infty$. Or define it to be $0$?}
%%%
%%%\reminder{Do we need to require $\sum_j L_j \geq K$? That is loading capacity greater than user number.}

In this section, we take proportional fairness among user achievable rates into consideration. 
%Since a user's achievable rate is the sum of rate from all the base stations, the 
The problem can be formulated as follows.
\begin{align}
\mbox{\textbf{P2-1:}} &\;\; \max_{\{x_{k_j}\}} \sum_{k=1}^K {\log\left(\sum_{j=1}^J x_{k_j}c_{k_j}\right)} \label{eqn:mmhn91} \\
\mbox{s.t.} &\;\; \mbox{same constraints as problem {\bf P1-2}}. \nonumber
%\mbox{s.t.} &\;\; \sum_{k=1}^K x_{k_j} \leq L_j, j=1,2,\cdots,J \nonumber \\
%            &\;\; \sum_{j=1}^J x_{k_j} \leq 1, k=1,2,\cdots,K \nonumber \\
%            &\;\; \mbox{Constraints (\ref{eqn:mmhn95})~(\ref{eqn:mmhn80})}. \nonumber
\end{align}

Problem~\textbf{P2-1} is a nonlinear integer programming problem, which is generally NP-hard. 
%NP-complete as well. 
To get a better understanding of the problem, we examine its equivalent problem as follows.
\begin{align}
\mbox{\textbf{P2-2:}} &\;\; \max_{x_{k_j}} \prod_{k=1}^K{\left(\sum_{j=1}^J x_{k_j}c_{k_j}\right)} \label{eqn:mmhn74} \\
\mbox{s.t.} &\;\; \mbox{same constraints as problem {\bf P1-2}}. \nonumber
%\mbox{s.t.} &\;\; \sum_{k=1}^K x_{k_j} \leq L_j, j=1,2,\cdots,J \nonumber \\
%            &\;\; \sum_{j=1}^J x_{k_j} \leq 1, k=1,2,\cdots,K \nonumber \\
%            &\;\; \mbox{Constraints (\ref{eqn:mmhn95})~(\ref{eqn:mmhn80})}. \nonumber
\end{align}

Problem \textbf{P2-2} is a geometric programming problem, with binary variables. The objective function is a posynomial function with $J^K$ terms. Conventionally, to solve geometric programming problems we need to introduce new variables such as $y = \log(x)$ so that geometric programming can be solved via convex programming. However, here $x_{k_j}$'s are binary. Since $\log(0)=-\infty$, we could not apply these techniques. 
Another heuristic scheme is to firstly sort these $J^K$ coefficients, and then find $L_j$ maximal coefficients for each BS. However, even sorting these $J^K$ coefficients could be computationally prohibitive even for a small system, 
%area with just $10$ BS's and $100$ users, 
which requires $\mathcal{O}(J^K \log (J^K))$ operations.

A key observation about 
%$\log()$ 
the logarithm function is that $\log(\sum_i{\tau_i}) \leq \sum_i \log(\tau_i)$, for all $\tau_i \geq 2$. Therefore, in practice,\footnote{Recall that $c_{k_j}$ is the achievable rate of user $k$ connecting to BS $j$. $c_{k_j} \geq 2$ is generally satisfied in current wireless systems with a sufficiently large bandwidth and high transmission power.} 
the optimal value of problem \textbf{P2-1} is upper bounded by that of the following problem.
\begin{align}
\mbox{\textbf{NP2:}} &\;\; \max_{x_{k_j}} \sum_{k=1}^K \sum_{j=1}^J x_{k_j}{\log(c_{k_j})} \label{eqn:mmhn71} \\
\mbox{s.t.} &\;\; \mbox{same constraints as problem {\bf P1-2}}. \nonumber
%\mbox{s.t.} &\;\; \sum_{k=1}^K x_{k_j} \leq L_j, j=1,2,\cdots,J  \nonumber \\
%            &\;\; \sum_{j=1}^J x_{k_j} \leq 1, k=1,2,\cdots,K \nonumber \\
%            &\;\; \mbox{Constraints (\ref{eqn:mmhn95}),~(\ref{eqn:mmhn80})}. \nonumber
\end{align}

%However, we have the following proposition.
We have the following results for the transformed problems.

\begin{lemma}
Problems~\textbf{P2-1} and~\textbf{NP2} are equivalent.
\label{prop222}
\end{lemma}
\begin{proof}
Recall that if $\eta_k=0$, we define $\mathcal{U}(\eta_k)=0$.
%$\mathcal{U}(\eta_k)=\log(\eta_k)=0$.
The second constraint $\sum_{j=1}^J x_{k_j} \leq 1$ imposes that each user could only connect to one BS. Consequently, $\sum_j x_{k_j}{\log(c_{k_j})} = \log(\sum_j x_{k_j}c_{k_j})$. Furthermore, we have $\sum_k \sum_j x_{k_j}{\log(c_{k_j})} = \sum_k{\log(\sum_j x_{k_j}c_{k_j})}$.
\end{proof}

\smallskip
Comparing problems \textbf{NP2} to \textbf{P1-2}, we find they are actually equivalent. Thus we can obtain the optimal value of \textbf{P2-1} by applying the same technique used to solve problem \textbf{P1-2}. 
%\textbf{NP1}. 
%An important conclusion stated in Corollary \ref{coro1} readily follows.
We then have the following lemma.
%corollary. 

\begin{lemma}
Sum rate maximization in Section \ref{subsec:cms} also achieves proportional fairness.
\label{coro1}
\end{lemma}

%%%\reminder{Is it because everyone can be connected? So no fairness needed.}

%\reminder{The upper bound is kind of trivial. We may delete it.}
%Next we derive an upper bound of the optimal value of problem \textbf{P2-1}. 
\begin{lemma}
The optimal value of problem \textbf{P2-1} is upper bounded by $UB_1=\sum_{k=1}^K \max_{j} \log_2(c_{k_j})$.
\label{prop1}
\end{lemma}
\begin{proof}
Denote $m=\max{\left\{\ln(c_{k_1}),\ln(c_{k_2}),\cdots,\ln(c_{k_J})\right\}}$, we have
\begin{eqnarray}
\log_2\left(\sum_{j=1}^J x_{k_j}c_{k_j}\right) &\leq& \log_2\left(\sum_{j=1}^J \frac{e^m}{e^m} e^{\ln(c_{k_j})}\right) \nonumber \\
&=&\log_2(e^m) + \log_2\left(\sum_{j=1}^J e^{\ln(c_{k_j})-m}\right) \nonumber \\
&\leq&m\log_2(e)+\log_2(J)
\label{eqn:mmhn73}
\end{eqnarray}
The first inequality is because $x_{k_j} \leq 1$.
The second inequality is due to the fact that $m$ is the largest one among all the $x_{k_j}c_{k_j}$ and $e^{\ln(x_{k_j}c_{k_j})-m} \leq 1$.

On the other hand, it follows the constraint $\sum_{j=1}^J x_{k_j} \leq 1$ that
%\begin{equation}
$\log_2\left(\sum_{j=1}^J x_{k_j}c_{k_j}\right) \leq \log_2(e^m)$, 
%\label{eqn:mmhn72}
%\end{equation}
%
%Since (\ref{eqn:mmhn72}) is 
which is a better bound than (\ref{eqn:mmhn73}). We thus have $UB_1=\sum_{k=1}^K \max_{j} \log_2(c_{k_j})$.
%%%\reminder{Note that since we have $\sum_j L_j \geq K$. $\sum_{j=1}^J{x_{k_j}=1}$. That is a user will be served by a base station.}
\end{proof}

\smallskip
For comparison purpose, we propose two sub-optimal greedy algorithms, i.e., Algorithms~\ref{alg:gafmm1} and~\ref{alg:gafmm2}, as benchmarks.
%The above two greedy algorithms 
They can be directly used for comparison with problem \textbf{P1-1}.
To compare with problem \textbf{P2-1}, in Algorithm~\ref{alg:gafmm1}, we need to change Steps $7$ and $8$ as
``\textbf{while} $\exists j,~L_j \neq 0$ $~\&~\max_{k,j} \log(c_{k,j})>0$ \textbf{do}'' and ``Find $(k^*,j^*)=\argmax_{k,j} \{ \log(c_{k,j}) \}$,'' respectively. 
In Algorithm \ref{alg:gafmm2}, we need to change Step $8$ and $9$ as
``\textbf{while} $L_{j} \neq 0$ $~\&~\max_{k} \log(c_{k,j})>0$'' and ``Find $(k^*,j)=\argmax_{k} \log(c_{k,j})$,'' respectively. 

\begin{algorithm} [!t]%[H]
\small
\SetAlgoLined
	Initialize $\mathcal{K}=\left\{1,2,\cdots,K\right\}$, $L_j,\forall j \in \mathcal{J}$ and $x_{k_j}$ to be an all-zero matrix \;
	\For{$k=1$ to $K$}{
		\For{$j=1$ to $J$}{
			Compute $c_{k_j}$ as in~(\ref{eqn:mmhn80}) \;
			}
		}
	\While{$\exists j,~L_j \neq 0$}{
			Find $(k^*,j^*)=\argmax_{k,j} \{ c_{k,j} \}$ \;
			\If{$L_{j^*} \neq 0$}{
				$x_{k^{*}_{j^*}} = 1$ \;
				$L_{j^*} = L_{j^*} - 1$ \;
				$\mathcal{K}=\mathcal{K} \backslash k^*$ \;
			}
			}
\caption{Greedy Algorithm 1 for User Association}
\label{alg:gafmm1}
\end{algorithm}

\begin{algorithm} [!t]%[H]
\small
\SetAlgoLined
	Initialize $\mathcal{K}=\left\{1,2,\cdots,K\right\}$, $L_j,\forall j \in \mathcal{J}$ and $x_{k_j}$ to be an all-zero matrix \;
	\For{$k=1$ to $K$}{
		\For{$j=1$ to $J$}{
			Compute $c_{k_j}$ as in~(\ref{eqn:mmhn80}) \;
			}
		}
	\For{$j=1$ to $J$}{
			\While{$L_{j} \neq 0$}{
				Find $(k^*,j)=\argmax_{k} \{ c_{k,j} \}$ \;
				$x_{k^{*}_{j}} = 1$ \;
				$L_{j} = L_{j} - 1$ \;
				$\mathcal{K}=\mathcal{K} \backslash k^*$ \;
			}
			}
\caption{Greedy Algorithm 2 for User Association}
\label{alg:gafmm2}
\end{algorithm}

%%%
%%%\reminder{Does BSs have priority?}
%%%
%%%\reminder{Or each BS chooses the best $L_j$ users?}
%%%
%%%\reminder{prove a bound for this}
%%%
%%%\reminder{Special case: only $J$ users, each BS only serve one user}

%-----------------------------------------------------------
\subsection{Joint Resource Allocation and User Association} \label{subsec:jraua}
%-----------------------------------------------------------

In this section, we take resource allocation into account.
Consider a massive MIMO OFDMA HetNet.
%The system resources contain time and frequency. 
In OFDMA systems, such as LTE, the time-frequency resource is divided into resource blocks (RB). A typical RB consists of $12$ subcarriers ($180$kHz) in the frequency domain and $7$ OFDMA symbols in the time domain ($0.5$ ms). So the system may have up to several hundreds of RBs. We normalize it to be a unit number.
A user $k$ connecting to a BS $j$ gets a portion $\beta_{k_j}$ of the overall resource.
The goal is to maximize the system utility considering both resource allocation and user association.
%Note, $\log()$ rate utility is considered here. And 

Considering the logarithm rate utility and defining $\Phi_j = \left\{k \;|\; x_{k_j} = 1 \right\}$, the problem is formulated as follows.
\begin{align}
\mbox{\textbf{P3-1:}} &\;\; \max_{\{x_{k_j},\beta_{k_j}\}} \sum_{k=1}^K{\log\left(\sum_{j=1}^J x_{k_j}c_{k_j} \beta_{k_j} \right)} \label{eqn:mmhn63} \\
\mbox{s.t.} &\;\; \sum_{k \in \Phi_j} \beta_{k_j} \leq 1, \; j=1,2,\cdots,J \nonumber \\
&\;\; \mbox{same constraints as problem {\bf P1-2}}. \nonumber 
%\mbox{s.t.} &\;\;   \sum_{k=1}^K x_{k_j} \leq L_j, j=1,2,\cdots,J \nonumber \\
%&\;\; \sum_{j=1}^J x_{k_j} \leq 1, k=1,2,\cdots,K \nonumber \\
%&\;\; \sum_{k \in \Phi_j} \beta_{k_j} \leq 1, j=1,2,\cdots,J \nonumber \\
%&\;\; \mbox{Constraints (\ref{eqn:mmhn95}),~(\ref{eqn:mmhn80})}. \nonumber
\end{align}
%Note that in (\ref{eqn:mmhn63}), $\Phi_j$ is the user set defined as:
%\begin{eqnarray}
%\label{eqn:mmhn63nn}
%\Phi_j = \left\{k | x_{k_j} = 1 \right\}.
%\end{eqnarray}
To solve problem 
%(\ref{eqn:mmhn63}) contains two levels of coupled question: 
{\bf P3-1}, we need to: (i) select users for each BS to serve and (ii) allocate resources to the associated users at each BS. We next propose a series of primal decomposition and dual decomposition to solve the problem optimally. 

%It is worth pointing out that another possible formulation of the problem is as follows.
%\begin{eqnarray}
%\max_{x_{k_j},\beta_{k_j}} && \sum_{k=1}^K{\log\left(\sum_{j=1}^J x_{k_j}c_{k_j} \beta_{k_j} \right)} \label{eqn:mmhn63eq} \\
%\mbox{s.t.} &&   \sum_{k=1}^K x_{k_j} \leq L_j, j=1,2,\cdots,J \nonumber \\
%&& \sum_{j=1}^J x_{k_j} = 1, k=1,2,\cdots,K \nonumber \\
%&& \sum_{k \in \Phi_j} \beta_{k_j} \leq 1, j=1,2,\cdots,J \nonumber \\
%&& \mbox{Constraints (\ref{eqn:mmhn95}),~(\ref{eqn:mmhn80})}. \nonumber
%\end{eqnarray}

It is worth noting that the problem can also be formulated in a different way, by substituting constraint $\sum_{j=1}^J x_{k_j} \leq 1, k=1,2,\cdots,K$ with a new constraint $\sum_{j=1}^J x_{k_j} = 1, k=1,2,\cdots,K$. We call this problem {\bf P3-2}.  
Comparing these two formulations, we have the following observations. 
\begin{enumerate}
	\item Problem %(\ref{eqn:mmhn63}) 
	{\bf P3-1} does not require that every user must be connected, while problem {\bf P3-2} requires each user be connected, even under unfavorable conditions.  
	%However, problem (\ref{eqn:mmhn63eq}) does require all the users to be connected even under unfavorable conditions.
	\item Problem {\bf P3-2} %(\ref{eqn:mmhn63eq}) 
	has a more stringent requirement than problem {\bf P3-1}. %(\ref{eqn:mmhn63}). 
	Therefore the optimal value of problem {\bf P3-2}
	%(\ref{eqn:mmhn63}) 
	is upper bounded by that of problem {\bf P3-1}. %(\ref{eqn:mmhn63}).
	\item Since problem {\bf P3-1}
	%(\ref{eqn:mmhn63}) 
	offers more choices of user association, problem {\bf P3-1}
	%(\ref{eqn:mmhn63}) 
	is slower in convergence than problem {\bf P3-2}. 
	%(\ref{eqn:mmhn63eq}).
\end{enumerate}

We focus on the harder problem {\bf P3-1}. 
%(\ref{eqn:mmhn63}). 
Given the algorithm to solve problem {\bf P3-1}, problem {\bf P3-2} can be readily solved. 
%(\ref{eqn:mmhn63}), the solution to problem (\ref{eqn:mmhn63eq}) is ready to obtain.
Due to integer variables $x_{k_j}$ and real variables $\beta_{k_j}$, problem {\bf P3-1}
%(\ref{eqn:mmhn63}) 
is a mixed integer nonlinear programming problem (MINLP), which is generally NP-hard. However, next we propose an algorithm to obtain its optimal solution.

Since $x_{k_j}$'s take binary values and $\sum_{j=1}^J x_{k_j} \leq 1$, we have
%\begin{equation}
%\label{eqn:mmhn62}
$\sum_{k=1}^K{\log\left(\sum_{j=1}^J x_{k_j}c_{k_j} \beta_{k_j} \right)}$ $=$ $\sum_{k=1}^K\sum_{j=1}^J x_{k_j}$${\log(c_{k_j} \beta_{k_j} )}$. 
%\end{equation}
Recall that if $\sum_{j=1}^J x_{k_j}=0$, the logarithmic utility is $0$.
%Also note that the choices of $\beta_{k_j}$ rely on the values of $x_{k_j}$. 
Thus problem {\bf P3-1} can be reformulated as
% (\ref{eqn:mmhn63}) is in fact the following:
\begin{align}
\mbox{\textbf{P3-3:}} &\;\; \max_{\{x_{k_j},\beta_{k_j}\}} \sum_{k=1}^K\sum_{j=1}^J x_{k_j}{\log(c_{k_j} \beta_{k_j} )} \label{eqn:mmhn61} \\
\mbox{s.t.} &\;\; \mbox{same constraints as problem {\bf P3-1}}. \nonumber
%\mbox{s.t.} &&   \sum_{k=1}^K x_{k_j} \leq L_j, j=1,2,\cdots,J \nonumber \\
%&& \sum_{j=1}^J x_{k_j} \leq 1, k=1,2,\cdots,K \nonumber \\
%&& \sum_{k \in \Phi_j} \beta_{k_j}(x_{k_j}) \leq 1, j=1,2,\cdots,J \nonumber \\
%&& \mbox{Constraints (\ref{eqn:mmhn95})~(\ref{eqn:mmhn80})}. \nonumber
\end{align}

The choices of $\beta_{k_j}$ rely on the values of $x_{k_j}$. Given these coupled variables, we first apply the Primal Decomposition method~\cite{chiang2006} to decompose problem {\bf P3-3}
%(\ref{eqn:mmhn61}) 
to the following two levels of problems. Fixing variables $x_{k_j}$'s, we have the \textit{lower level problem} as
\begin{align}
\max_{\{\beta_{k_j}\}} &\;\; \sum_{k=1}^K\sum_{j=1}^J x_{k_j}{\log(c_{k_j} \beta_{k_j} )} \label{eqn:mmhn60} \\
\mbox{s.t.} &\;\; \sum_{k \in \Phi_j} \beta_{k_j} \leq 1, j=1,2,\cdots,J. \nonumber 
\end{align}
When the $\beta_{k_j}$'s are fixed, the \textit{higher level problem} (or, the \textit{master problem}) is given by
\begin{align}
\max_{\{x_{k_j}\}} &\;\; \sum_{k=1}^K\sum_{j=1}^J x_{k_j}{\log(c_{k_j} \beta_{k_j} )} \label{eqn:mmhn59} \\
\mbox{s.t.} &\;\; \mbox{same constraints as problem {\bf P1-2}}. \nonumber
%\mbox{s.t.} &&   \sum_{k=1}^K x_{k_j} \leq L_j, j=1,2,\cdots,J \nonumber \\
%&& \sum_{j=1}^J x_{k_j} \leq 1, k=1,2,\cdots,K \nonumber \\
%&& \mbox{Constraints (\ref{eqn:mmhn95}),~(\ref{eqn:mmhn80})}. \nonumber
\end{align}
%where $\beta_{k_j}$ are fixed.
Since there are no couplings among the subproblems, the lower level problem~(\ref{eqn:mmhn60}) can be further decomposed into $L$ subproblems as follows.
\begin{align}
\max_{\{\beta_{k_j}\}} &\;\; \sum_{k=1}^K x_{k_j}{\log(c_{k_j} \beta_{k_j} )} \label{eqn:mmhn58} \\
\mbox{s.t.} &\;\; \sum_{k \in \Phi_j}^K \beta_{k_j} \leq 1, j=1,2,\cdots,J. \nonumber 
\end{align}

Defining Lagrange multiplier $\lambda$, the Lagrangian of problem~(\ref{eqn:mmhn58}) is defined as
\begin{equation}
\mathcal{L} = \sum_{k=1}^K x_{k_j}{\log(c_{k_j} \beta_{k_j} )} + \lambda \left(1-\sum_{k=1}^K \beta_{k_j}\right). \label{eqn:mmhn57} 
\end{equation}
%where $\lambda$ is the Lagrange multiplier. 
Applying KKT conditions \cite{boydcvx}, the optimal solution can be obtained as follows. 
\begin{equation}
\beta_{k_j} = \frac{x_{k_j}}{\sum_{k=1}^K x_{k_j}}. \label{eqn:mmhn56} 
\end{equation}

Substituting~(\ref{eqn:mmhn56}) into the master problem, the objective function becomes 
\begin{equation} \label{eqn:mmhn55} 
\sum_{k=1}^K\sum_{j=1}^J x_{k_j}{\log\left( \frac{c_{k_j}}{\sum_{k=1}^K x_{k_j}} \right)}.
\end{equation}
%\begin{align}
%\max_{x_{k_j}} &\;\; \sum_{k=1}^K\sum_{j=1}^J x_{k_j}{\log\left( \frac{c_{k_j}}{\sum_{k=1}^K x_{k_j}} \right)} \label{eqn:mmhn55} \\
%\mbox{s.t.} &\;\; \mbox{same constraints as problem {\bf P1-2}}. \nonumber
%\mbox{s.t.} &&   \sum_{k=1}^K x_{k_j} \leq L_j, j=1,2,\cdots,J \nonumber \\
%&& \sum_{j=1}^J x_{k_j} \leq 1, k=1,2,\cdots,K \nonumber \\
%&& \mbox{Constraints (\ref{eqn:mmhn95}),~(\ref{eqn:mmhn80})}. \nonumber
%\end{align}
Note that we have dropped one $x_{k_j}$ term in~(\ref{eqn:mmhn55}), since due to the definition~(\ref{eqn:mmhn95}), we have $(x_{k_j})^2=x_{k_j}$.
Since $\sum_{k=1}^K x_{k_j}$ is in the denominator, problem~(\ref{eqn:mmhn55}) has coupled objectives. The main idea of addressing the coupled objective is to introduce auxiliary variables and additional equality constraints so that the coupling in the objective function is transferred to coupling in the constraint~\cite{chiang2006}.
We thus introduce a new variable, which is defined as:
\begin{equation}
\Xi_j=\sum_{k=1}^K x_{k_j}. \label{eqn:mmhn54} 
\end{equation}

To solve the above problem, we relax $x_{k_j}$ to a real number in $[0, 1]$. However, we will show later that even if we have relaxed the variables, we could still find the optimal solution to the original problem. The relaxed problem to be solved is
\begin{align}
\max_{\{x_{k_j}\}} &\;\; \sum_{k=1}^K\sum_{j=1}^J x_{k_j}{\log\left( \frac{c_{k_j}}{\Xi_j} \right)} \label{eqn:mmhn53} \\
\mbox{s.t.} &\;\; \Xi_j \leq L_j, \; j=1,2,\cdots,J \nonumber \\
&\;\; \sum_{j=1}^J x_{k_j} \leq 1, \; k=1,2,\cdots,K \nonumber \\
&\;\; 0 \leq x_{k_j} \leq 1, \mbox{for all} \; k, j \nonumber \\
&\;\; \mbox{Constraints (\ref{eqn:mmhn80}),~(\ref{eqn:mmhn54})}. \nonumber
\end{align}

Problem~(\ref{eqn:mmhn53}) is a convex optimization problem. Defining Lagrange multipliers for the equality constraints~(\ref{eqn:mmhn54}), problem~(\ref{eqn:mmhn53}) can be solved with the dual decomposition method.
Alternatively, we propose Algorithm~\ref{alg:tldda} to obtain the optimal solution of problem (\ref{eqn:mmhn53})~\cite{Jeff13,yixuicc,yiaccess}. In Algorithm~\ref{alg:tldda}, $\delta^{(t)}$ is the step size at the $t$-th iteration given by
\begin{eqnarray}
\label{eqn:mmhn52}
\delta^{(t)} = \frac{\vartheta}{t+\gamma},
\end{eqnarray}
where $\vartheta$ and $\gamma$ are positive numbers.

\begin{theorem}
Algorithm \ref{alg:tldda} optimally solves problem~(\ref{eqn:mmhn53}).
\end{theorem}
\begin{proof}
Let $\mathbf{x}_k^{(t)}$ denote the solution produced by Algorithm~\ref{alg:tldda} at step $t$.
Let $\partial \mathcal{U}(\mathbf{x}_k^{(t)})$ be the subgradient of the objective function in problem~(\ref{eqn:mmhn53}) at step $t$. It can be easily verified that the updated direction in step~\ref{stepGD}  
%$13$ 
of Algorithm~\ref{alg:tldda} is the subgradient direction. Since $\Xi_j$ is upper bounded by $L_j$ and $K$, and $\sum_{k=1}^K x_{k_j}$ is upper bounded by $K$, $\partial \mathcal{U}(\mathbf{x}_k^{(t)})$ is also bounded. 

Denote $\mathcal{U}_a$ as the final result produced by Algorithm \ref{alg:tldda} and $\mathcal{U}^*$ as the optimal solution of problem (\ref{eqn:mmhn53}). We prove the theorem by contradiction. Assume that $\mathcal{U}_a$ is not optimal. Then there must exist an $\epsilon>0$ such that
$\mathcal{U}_a + 2\epsilon < \mathcal{U}^*$. 
%\begin{eqnarray}
%\label{eqn:mmhn49}
%\mathcal{U}_a + 2\epsilon < \mathcal{U}^*.
%\end{eqnarray}
Then there must be a solution $\mathbf{\hat{x}}_k$ so that
\begin{eqnarray}
\label{eqn:mmhn48}
\mathcal{U}_a + 2\epsilon <\mathcal{U}(\mathbf{\hat{x}}_k).
\end{eqnarray}
Let $t_0$ be sufficiently large so that for any $t>t_0$ we have
\begin{eqnarray}
\label{eqn:mmhn47}
\mathcal{U}(\mathbf{x}_k^{(t)}) \leq \mathcal{U}_a + \epsilon.
\end{eqnarray}
Combining~(\ref{eqn:mmhn48}) and~(\ref{eqn:mmhn47}), we have
$\mathcal{U}(\mathbf{x}_k^{(t)}) + \epsilon < \mathcal{U}(\mathbf{x}_k^{(t)})$.
%\begin{eqnarray}
%\label{eqn:mmhn46}
%\mathcal{U}(\mathbf{x}_k^{(t)}) + \epsilon < \mathcal{U}(\mathbf{x}_k^{(t)}).
%\end{eqnarray}

Let $\kappa$ be a positive number that satisfies $\kappa \leq \inf \left\{ \|\partial \mathcal{U}(\mathbf{x}_k^{(t)})\|\right\}$, for all $t$.
It follows that
\begin{align}
\label{eqn:mmhn50}
&\; \| \mathbf{x}_k^{(t+1)} - \mathbf{\hat{x}}_k\|^2 %\\
%=&\; 
= \| \mathbf{x}_k^{(t)} -\delta^{(t)}\partial \mathcal{U}^{(t)}- \mathbf{\hat{x}}_k\|^2 \\ %\nonumber \\
=&\; \| \mathbf{x}_k^{(t)} - \mathbf{\hat{x}}_k\|^2 + (\delta^{(t)})^2 \| \partial \mathcal{U}^{(t)} \|^2 - 
2\delta^{(t)}(\partial \mathcal{U}^{(t)})^H(\mathbf{x}_k^{(t)}-\mathbf{\hat{x}}_k)
\nonumber \\
%&\; 2\delta^{(t)}(\partial \mathcal{U}^{(t)})^H(\mathbf{x}_k^{(t)}-\mathbf{\hat{x}}_k) \nonumber \\
\geq&\; \| \mathbf{x}_k^{(t)} - \mathbf{\hat{x}}_k\|^2 + (\delta^{(t)})^2 \| \partial \mathcal{U}^{(t)} \|^2 -2\delta^{(t)}(\mathcal{U}(\mathbf{x}_k^{(t)})-\mathcal{U}(\mathbf{\hat{x}}_k)) \nonumber \\
\geq&\; \| \mathbf{x}_k^{(t)} - \mathbf{\hat{x}}_k\|^2 + (\delta^{(t)})^2 \kappa^2 +2\delta^{(t)}\epsilon \nonumber \\
\geq&\; \| \mathbf{x}_k^{(t)} - \mathbf{\hat{x}}_k\|^2 +2\delta^{(t)}\epsilon 
%\nonumber \\
%\geq&\; 
\geq \cdots \nonumber \\
\geq&\; \| \mathbf{x}_k^{(t_0)} - \mathbf{\hat{x}}_k\|^2 +2\epsilon \sum_{j=t_0}^t \delta^{(j)}. \nonumber 
\end{align}
Note that the first inequality is due to the property of subgradient. So we finally have $\| \mathbf{x}_k^{(t+1)} - \mathbf{\hat{x}}_k\|^2 \geq \| \mathbf{x}_k^{(t_0)} - \mathbf{\hat{x}}_k\|^2 +2\epsilon \sum_{j=t_0}^t \delta^{(j)}$, which cannot hold for sufficiently large $t$. Thus Algorithm \ref{alg:tldda} optimally solves problem (\ref{eqn:mmhn53}).
\end{proof}

\begin{algorithm} [!t]
\small
\SetAlgoLined
$t=0$, $\mathbf{\lambda}^{(1)}=0$ \;
\While{not converged}
{
$t\leftarrow t+1$ \;
\For {$k=1, \cdots, K$}{
\For {$j=1, \cdots, J$}{
Compute  $c_{k_j}$ as in (\ref{eqn:mmhn80}) \;}
Find $j^*=\argmax_j \left\{\log(c_{k_j}-\lambda_j^{(t)})\right\}$ \;
Let $x_{k_j}^{(t)}=0$ for $j\neq j^*$ \;
\eIf{$\log(c_{k_j}-\lambda_j^{(t)}) \geq 0$}{
	$x_{k_j^*}^{(t)}=1$ \;
}{
	$x_{k_j^*}^{(t)}=0$ \;
}
%Let $x_{k_j^*}^{(t)}=1$, if $\log(c_{k_j}-\lambda_j^{(t)}) \geq 0$ \;
%Otherwise, $x_{k_j^*}^{(t)}=0$ \;
}
\For{$j=1,\cdots,J$}{ \label{stepGD} 
Each BS chooses a step size $\delta^{(t)}$ and computes
$\Xi_j^{(t+1)}=\mbox{min}\{L_j,e^{(\lambda_j^{(t)}-1)}\}$ and 
$\lambda_j^{(t+1)}=\lambda_j^{(t)}-\delta^{(t)}(\Xi_j^{(t)}-\sum_{k=1}^K x_{k_j}^{(t)})$ \; 
% then annouces them to the system;
}
}
\caption{Two Layer Dual Decomposition Algorithm for Optimization Problem~(\ref{eqn:mmhn53})}
\label{alg:tldda}
\end{algorithm}

\begin{theorem} 
The optimal solution to problem~(\ref{eqn:mmhn53}) is also feasible and optimal to problem (\ref{eqn:mmhn55}).
\label{themo}
\end{theorem}

\begin{proof}
From problem~(\ref{eqn:mmhn53}) to~(\ref{eqn:mmhn55}), we relax the variables from binary to real and introduce an equality constraint. The equality constraint does not change the problem. 
%is just another representation of the problem and changes nothing. 
So the optimal value to problem~(\ref{eqn:mmhn55}) provides an upper bound to that of problem~(\ref{eqn:mmhn53}). 
However, it can be observed from Algorithm~\ref{alg:tldda} that the solutions to problem~(\ref{eqn:mmhn55}) are integers rather than fractions. 
So the solutions are also feasible to problem~(\ref{eqn:mmhn53}).   
%So we could set the solutions to problem (\ref{eqn:mmhn53}) exactly the same as the solutions to problem (\ref{eqn:mmhn55}). 
Since the solutions to problem~(\ref{eqn:mmhn53}) cannot result in a higher optimal value than the solutions to problem~(\ref{eqn:mmhn55}), the solutions to problem~(\ref{eqn:mmhn55}) are exactly the solutions to problem~(\ref{eqn:mmhn53}) as well. Henceforth, even though we transform problem~(\ref{eqn:mmhn53}) to problem~(\ref{eqn:mmhn55}), the optimal solution is not affected by the transformation.
\end{proof}

To sum up, the optimal solution to problem~(\ref{eqn:mmhn55}) can be solved with Algorithm~\ref{alg:tldda}. For comparison purpose, we also propose two greedy algorithms as benchmarks, which are presented in Algorithms~\ref{alg:gafmm3} and~\ref{alg:gafmm4}. The main idea of the greedy algorithms is to first identify the most desirable user-BS pair, and then to allocate all the resource to that user. This is repeated until convergence is reached.

\begin{algorithm} [!t]%[H]
\small
\SetAlgoLined
	Initialize $\mathcal{K}=\left\{1,2,\cdots,K\right\}$ and $\mathcal{J}=\left\{1,2,\cdots,J\right\}$ $x_{k_j}$ to be an all-zero matrix \;
	\For{$k=1$ to $K$}{
		\For{$j=1$ to $J$}{
			Compute $c_{k_j}$ as in (\ref{eqn:mmhn80}) \;
			}
		}
	\While{$\max_{k,j} \log(c_{k,j})>0$}{
			Find $(k^*,j^*)=\argmax_{k,j} \log(c_{k,j})$ \;
			$x_{k^{*}_{j^*}} = 1$ \;
			$\mathcal{K}=\mathcal{K} \backslash k^*$ \;
			$\mathcal{J}=\mathcal{J} \backslash J^*$ \;
			}
\caption{Greedy Algorithm 4 for Joint Resource Allocation and User Association}
\label{alg:gafmm3}
\end{algorithm}

\begin{algorithm} [!t]%[H]
\small
\SetAlgoLined
	Initialize $\mathcal{K}=\left\{1,2,\cdots,K\right\}$, $\mathcal{J}=\left\{1,2,\cdots,J\right\}$ and $x_{k_j}$ to be an all-zero matrix \;
	\For{$k=1$ to $K$}{
		\For{$j=1$ to $J$}{
			Compute $c_{k_j}$ as in (\ref{eqn:mmhn80}) \;
			}
		}
	\For{$j=1$ to $J$}{
			\If{$\max_{k} \log(c_{k,j}) > 0$}{
				Find $(k^*,j)=\argmax_{k} c_{k,j}$ \;
				$x_{k^{*}_{j}} = 1$ \;
				$\mathcal{K}=\mathcal{K} \backslash k^*$ \;
			}
			}
\caption{Greedy Algorithm 5 for Joint Resource Allocation and User Association}
\label{alg:gafmm4}
\end{algorithm}

%-----------------------------------------------------------
\section{Distributed User Association} \label{sec:dist}

In the previous section, we assume a central controller that has global information and assigns users to the BS's. In this section, we consider distributed user association.
% where there is no central controller. 
We still assume that the BS's have all the CSI via uplink training. We further assume that all the BS's, including the massive MIMO microcell BS and the small cell BS's, belong to the same service provider. Each user makes its own decision based on the broadcast and local information.
Throughout this section, we do not allow fractional connection. We omit constraint~(\ref{eqn:mmhn95}) in the problem formulation, which is, however, enforced when solving the problem.

We model the behavior and interactions among the service provider and users using repeated game theory. The first key problem is to determine whether the game will converge. The second key problem is to analyze whether both sides are satisfactory about the outcome of the game, i.e., existence of the Nash Equilibrium.

%-----------------------------------------------------------
%\subsection{Service Provider's Strategy for Pricing} \label{subsec:spdtp}
%\subsection{User Association Game Formulation} \label{subsec:game}
\subsection{Service Provider Sets the Price} \label{subsec:spdtp}
%-----------------------------------------------------------

The players of the repeated game include the service provider and the users. 
%In the first subsection, we consider the case that during 
During each round of the game, the service provider determines the price of the connection service. The users decide whether or not to connect, and if to connect, to which BS. 
The strategy of the service provide is to set the price $p_{k_j}$ of each BS $j$ for each user $k$, while the strategy of each user $k$ is to set $x_{k_j}$ to either $0$ or $1$ for $j \in \mathcal{J}$.

%%%%%%%-----------------------------------------------------------
%%%%%%\subsubsection{User Discrimination} \label{subsubsec:distwpd}
%%%%%%%-----------------------------------------------------------

The utility of the service provider is defined as $\mathcal{U}_B=\sum_{k=1}^K \sum_{j=1}^J x_{k_j} p_{k_j}$. Since each BS is constrained by its maximum load capacity $L_j$, 
the service provider aims to solve the following problem. 
\begin{align}
\max_{\{p_{k_j}\}} &\;\; \mathcal{U}_B=\sum_{k=1}^K \sum_{j=1}^J x_{k_j} p_{k_j} \label{eqn:mmhn84} \\
\mbox{s.t.} &\;\; \sum_k x_{k_j} \leq L_j, \; j=1, 2, \cdots, J. \nonumber
\end{align}

The utility of each user is the data rate achieved minus its payment.
So each user aims to solve the following problem.
\begin{align}
\max_{\{x_{k_j}\}} &\;\; \mathcal{U}_k=\max \left\{ \omega_k \log\left(\sum_{j=1}^J x_{k_j}c_{k_j}\right) 
- \sum_{j=1}^J x_{k_j}p_{j_k}, 0 \right\} \label{eqn:mmhn83} \\
\mbox{s.t.} &\;\; \sum_j x_{k_j} \leq 1, \nonumber
\end{align}
where  
%used to trade-off the two objectives for user $k$. 
the logarithmic function represents the satisfaction level of a user $k$ towards its achievable rate, and $\omega_k$ is a weight used to tradeoff rate satisfaction and monetary payment. 
%$\omega_k$ is a weight factor used for conversion between the rate satisfaction and monetary measurement. 
We assume that the weight $\omega_k$ of each user is drawn from a finite set $\mathcal{W}$ with $\left| \mathcal{W} \right|$ elements. This assumption is true in real-world practice. For instance, $\$30$ for a wireless service with $60$ Mbps data rate is considered to be cheap; $\$45$ is considered to be reasonable; $\$60$ would be 
%thought as 
acceptable; $\$80$ would be expensive for most people; $\$100$ would be too expensive; and $\$150$ or above would not be an option for most people. So the weight of the users has generally finite choices of values based on common sense, and is typically in a range $=(0,W_{M})$, where $W_M$ is the maximum possible value for $\omega_k$. 

The repeated game is played as follows.
Initially, the service provider sets a price for each BS for each user and broadcasts the prices to the users. Knowing the prices, the users will feedback the service provider of their choices based on their own calculations. Then the service provider updates the prices and broadcasts them to the users. Users again inform the service provider of their choices, and so forth. The process is repeated until both the service provider and users are all satisfied with the price.

%-----------------------------------------------------------
%\subsection{Service Provider's Strategy for Pricing} \label{subsec:spdtp}
%\subsection{Nash Equilibrium and Convergence} \label{subsec:spdtp}
%-----------------------------------------------------------

Given the players, their strategies and utilities, we have the following definition for the NE of the user association game.  
%define the NE of the game as follows. 

\begin{definition} 
A strategy set $\left\{p^*_{k_j}, x^*_{k_j}\right\}$, for all $k$, $j$, is an NE of the repeated game if $\mathcal{U}_B(p^*_{k_j}, x^*_{k_j}) \geq \mathcal{U}_B(p_{k_j}, x^*_{k_j})$, for all $p_{k_j}$ and $\mathcal{U}_k(p^*_{k_j}, x^*_{k_j}) \geq \mathcal{U}_k(p^*_{k_j}, x_{k_j})$, for all $k$, $x_{k_j}$.
\end{definition}

Due to the constraint that each user can only connect to one BS, $\omega_k \log\left(\sum_{j=1}^J x_{k_j}c_{k_j}\right)=\sum_{j=1}^J x_{k_j} \omega_k \log(c_{k_j})$. 
%So each user locally optimize the following problem.
Therefore the objective function of problem~(\ref{eqn:mmhn83}) becomes
\begin{align}
%\max_{\{x_{k_j}\}} &\;\; 
\mathcal{U}_k=\max \left\{ \sum_{j=1}^J x_{k_j} \omega_k \log(c_{k_j}) - \sum_{j=1}^J x_{k_j}p_{j_k},0 \right\}. \label{eqn:mmhn83n} %\\
%\mbox{s.t.} &\;\; \sum_j x_{k_j} \leq 1. \nonumber
\end{align}
%
%Note that the constraint of 
For the reformulated problem~(\ref{eqn:mmhn83n}), the constraint $\sum_j x_{k_j} \leq 1$ indicates that 
%$\sum_j x_{k_j}$ could be $0$, i.e., 
a user may choose not to connect to any of the BS's. On the other hand, if we restrict $\sum_j x_{k_j}=1$, 
%i.e., a user must connect to a base station, 
then even if the service provider sets the prices to infinity, each user will still connect to a BS, which is clearly unreasonable.

Given the utility function~(\ref{eqn:mmhn83n}) and the constraint in~(\ref{eqn:mmhn83}), the optimal solution for each user can be derived as
\begin{align}
	& j^* = \argmax_{j \in \mathcal{J}} \left[\omega_k \log (c_{k_j}) - p_{j_k}\right] \\
	& x_{k_j} = \begin{cases}
    		1, & \text{if $j=j^*$ and } \omega_k \log (c_{k_j^*}) \geq p_{j^*_k} \\
    		0, & \text{otherwise}.
  				\end{cases}
\label{eqn:mmhn70}
\end{align}

Such users' decision can be interpreted this way. 
A user will choose the best connection based on its own evaluation.
If its evaluation of the connection is greater than or equal to the price, it will connect to this BS. Otherwise, the user will not connect to the BS. So we readily have the following result.

\begin{lemma} \label{lm:lm5}
The highest profit the service provider can obtain from a user $k$ towards BS $j$, is the user's evaluation.
\end{lemma}

The service provider aims to solve problem~(\ref{eqn:mmhn84}) by tuning variables $p_{k_j}$, $k=1,2,\cdots,K$, $j=1,2,\cdots,J$. However, the constraint $\sum_k x_{k_j} \leq L_j$ is implicitly coupled with all the 
%implicitly contains variable 
$p_{k_j}$'s, since according to the user's choice, $j^*=\argmax_{j \in \mathcal{J}} \left[\omega_k \log (c_{k_j}) - p_{j_k}\right]$. The service provider problem is actually with the following form. 
\begin{align}
\max_{\{p_{k_j}\}} &\;\; \mathcal{U}_B=\sum_{k=1}^K \sum_{j=1}^J x_{k_j} p_{k_j} \label{eqn:mmhn66} \\
\mbox{s.t.} &\;\; \sum_k x_{k_{j(p_{k_j})}} \leq L_j, \; j=1, 2, \cdots, J. \nonumber
\end{align}

%Thus, the optimization problem for the service provider is in fact (\ref{eqn:mmhn66}), which is mixed integer programming in nature and generally NP-hard to solve. 

Since problem~(\ref{eqn:mmhn66}) has coupling constraints, one may try to introduce Lagrange multipliers to the constraint and solve the resulting problem using dual decomposition. However, since $p_{k_j}$ is implicitly contained in the constraint, the gradient and subgradient are difficult to find. 
%So we need to resort to other ways. 
Next, we propose Algorithm~\ref{alg:disafsp1} for the service provider, and then prove that the algorithm achieves optimal utility for the service provider and the users.

\begin{theorem}
\label{prop10}
If the service provider adopts Algorithm \ref{alg:disafsp1}, the game converges and the NE can be achieved.
\end{theorem}
\begin{proof}
We first notice that the service provider has priority over the users. The users always make decisions based upon the service provider's price setting. Basically, the service provider controls when the repeated game terminates. 

In Algorithm~\ref{alg:disafsp1}, the service provider tests out the weight of each user using binary search with $\mathcal{O}(\log_2(\left| \mathcal{W} \right|))$ steps.
Once the service provider obtains $\omega_k$, $k=1,2,\cdots,K$, it then estimates the users' price evaluation matrix $\mathbf{V}$ as follows.
\begin{equation}
	v_{k_j} = c_{k_j}  \omega_k,
\label{eqn:mmhn69}
\end{equation}
where $v_{k_j}$ is the entry of matrix $\mathbf{V}$ at row $j$ and column $k$. 
%Then the service provider is ready to select users for the base stations by solving the following problem.
%Given the above fact, 
Following Lemma~\ref{lm:lm5}, the service provider can obtain its optimal price strategy by first selecting users for each BS and solving the following problem. 
\begin{align}
\max_{x_{k_j}} &\;\; \sum_{k=1}^K\sum_{j=1}^J {x_{k_j}v_{k_j}} \label{eqn:mmhn68} \\
\mbox{s.t.} &\;\; \sum_k x_{k_j} \leq L_j,~j=1,2,\cdots,J \nonumber \\
            &\;\; \sum_j x_{k_j} \leq 1,~k=1,2,\cdots,K \nonumber \\
            &\;\; \mbox{Constraints (\ref{eqn:mmhn95}),~(\ref{eqn:mmhn69})}. \nonumber
\end{align}

The optimal solution $x^*_{k_j}$ to the above problem can be solved in a similar way as solving problem~\textbf{P1-2}.
Then the optimal prices for the service provider can be obtained as follows.
\begin{equation}
  p^*_{k_j}=\begin{cases}
    v_{k_j}, & \text{if $x^*_{k_j}=1$};\\
    v_{k_j}+\epsilon, & \text{otherwise},
  \end{cases}
\label{eqn:mmhn67}
\end{equation}
where $\epsilon$ is an arbitrary positive number.

Therefore, by adopting Algorithm \ref{alg:disafsp1}, the optimal utility (highest) can be reached for the service provider. Meanwhile, we could see that all the users' utility must be $0$ due to the optimal price setting (i.e., each user's rate satisfaction matches its monetary payment). That means, all the users achieve the optimal utility given the price setting as well. Therefore, the game converges to the NE. 
%and the NE is met.
\end{proof}

\begin{algorithm} [!t]%[H]
\small
\SetAlgoLined
	Initialize $\omega_{MAX}$, $\omega_{MIN}$, $t=0$ \;
	\For{$k=1$ to $K$}{
		\For{$j=1$ to $J$}{
			Compute $c_{k_j}$ as in (\ref{eqn:mmhn80}) \;
			}
		}
		
	\For{$k=1$ to $K$}{
		$\omega_k^u(t) = \omega_{MAX}$ \;
		$\omega_k^l(t) = \omega_{MIN}$ \;
	}
	\While{not converged}{
	\For{$k=1$ to $K$}{
		$\hat{w}_{k}(t) = \frac{1}{2} (\omega_k^u(t)+\omega_k^l(t))$ \;
		\For{$j=1$ to $J$}{
			$p_{k_j}(t) = \max\left\{\hat{w}_{k}(t) \log(c_{k_j}),0\right\}$ \;
		}
	}
	$t \leftarrow t + 1$ \;
	\For{$k=1$ to $K$}{
	\uIf{$\left|F_k\right|>1$}{
		$\omega_k^u(t) = \omega_k^u(t-1)$ \;
		$\omega_k^l(t) = \omega_k^l(t-1)$ \;
	}
	\uElseIf{$\left|F_k\right|=1$}{
		$\omega_k^u(t) = \omega_k^u(t-1)$ \;
		$\omega_k^l(t) = \hat{\omega}_{k}(t)$ \;
	}
	\Else{
		$\omega_k^u(t) = \hat{\omega}_{k}(t)$ \;
		$\omega_k^l(t) = \omega_k^l(t-1)$ \;
	}
	}
	}
	
	\For{$k=1$ to $K$}{
		\For{$j=1$ to $J$}{
		Calculate $v_{k_j}$ as in (\ref{eqn:mmhn69}) using $\hat{w}_{k}$ \;
		}
	}
	Solve (\ref{eqn:mmhn68}) and find optimal price as in (\ref{eqn:mmhn67}) \;
\caption{Algorithm for Service Provider}
\label{alg:disafsp1}
\end{algorithm}

\smallskip
Note that it is 
%highly 
possible that the optimal utility of the service provider will be lower than the maximum utility during the game, because the load capacity constraint may be violated due to the distributed operation.

%%%%%%-----------------------------------------------------------
%%%%%\subsubsection{No User Discrimination} \label{subsubsec:distnopd}
%%%%%%-----------------------------------------------------------
%%%%%
%%%%%The service provider sets price of each base station. No user discrimination is allowed.
%%%%%In this way, users pay uniform price $p_j$ if them choose the same base station $j$.
%%%%%Users' strategy would be similar to that in (\ref{eqn:mmhn64}).
%%%%%
%%%%%\begin{eqnarray}
%%%%%j^* &=& \arg \max_{j \in \mathcal{J}} \left[w_k \log (c_{k_j}) - p_{j}\right], \\
%%%%%x_{k_j}&=&\begin{cases}
    %%%%%1, & \text{if $j=j^*$ and } w_k \log (c_{k_j^*}) \geq p_{j^*}.\\
    %%%%%0, & \text{otherwise} .
  %%%%%\end{cases}
%%%%%\label{eqn:mmhn64}
%%%%%\end{eqnarray}
%%%%%
%%%%%For the service provide, it could firstly use the binary search to test out the weight $w_k$ of each user as well. However, the remaining problem is as follows.
%%%%%
%%%%%\begin{eqnarray}
%%%%%\max_{p_{k_j}} && \mathcal{U}_B=\sum_{k=1}^K \sum_{j=1}^J x_{k_j} p_{j} \label{eqn:mmhn65} \\
%%%%%\mbox{s.t.} && \sum_k x_{k_{j(p_{j})}} \leq L_j \nonumber
%%%%%\end{eqnarray}

%-----------------------------------------------------------
%\subsection{Users' Stratege for Bidding} \label{subsec:ub}
\subsection{A User Bidding based Approach} \label{subsec:ub}
%-----------------------------------------------------------
%\reminder{not bid to a BS with satisfaction less than $0$}

We next consider a bidding approach to the problem. 
%It operates as follows. 
Before service starts, users bid to the service provider according to their predicted satisfaction towards each BS. And service provider determines whether or not to accept a user's bid and feedback the decisions to users. Then the users make another round of bids according to its predicted satisfaction and the service provider's decision history. The service provider again decides whether or not to accept a user's bid and feedback the decision, and so forth.

%Here we consider the 
Assume date-intensive users that
%users to be data-hungry and 
strive for as high data rate as possible.
%Considering that a user will not bid to a BS with satisfaction less than $0$, each 
Each user solves the following problem. 
\begin{align}
\max_{\{p_{k_j}\}} &\;\; \mathcal{U}_k=\max \left\{ \sum_{j=1}^J x_{k_j} \omega_k \log(c_{k_j}),0 \right\} \label{eqn:mmhn40} \\
\mbox{s.t.} &\;\; \sum_j x_{k_j(p_{k_j})} \leq 1. \nonumber
\end{align}
On the other hand, the service provide aims to maximize its utility, i.e., the total payment made by all the users. 
\begin{align}
\max_{\{x_{k_j}\}} &\;\; \mathcal{U}_B=\sum_{k=1}^K \sum_{j=1}^J x_{k_j} p_{k_j} \label{eqn:mmhn38} \\
\mbox{s.t.} &\;\; \sum_k x_{k_{j}} \leq L_j, \; j=1,2,\cdots,J. \nonumber
\end{align}
Note that the decision variables in these two problems are different from those in problems~(\ref{eqn:mmhn84}) and~(\ref{eqn:mmhn83}), respectively.  

%To sum up, the players of the game include the service provider and the users as well. The strategy of the user is to make payment $0$ or the satisfactory level to a BS. The strategy of the service provider is to determine whether or not accept a user's bid. The utility of each user is given in (\ref{eqn:mmhn40}) and the utility of the service provider is given in (\ref{eqn:mmhn38}).

We assume the general case that $K \geq \sum_{j=1}^J L_j$ (i.e., not all the users can be served). 
%The optimal solution for each user is simple. 
In order to achieve the greatest level of satisfaction, each user makes the highest possible payment. So the optimal solution for each user is
\begin{eqnarray}
	p_{k_j}=\max \left\{ \sum_{j=1}^J x_{k_j} \omega_k \log(c_{k_j}),0 \right\}. \label{eqn:mmhn39} 
\end{eqnarray}
The optimal strategy for the service provider is summarized in Algorithm~\ref{alg:distub}.

\begin{algorithm} [!t]%[H]
\small
\SetAlgoLined
	\While{not converged}{
	\For{$j=1$ to $J$}{
		\uIf{BS $j$ is bidden by $\leq L_j$ users}{
			Keep all the users in BS $j$'s waiting list \;
		}\Else{
			Keep the top $L_j$ users with the highest bids and reject the other users \;			
		}
	}
	}
\caption{Algorithm for the Service Provider with the Bidding Approach}
\label{alg:distub}
\end{algorithm}

During the first stage of the game, each user offers a price to its most desirable BS. Algorithm~\ref{alg:distub} is used to check if each BS $j$ receives more than $L_j$ bids. The service provider only puts $L_j$ top users on BS $j$'s waiting list based on the offered prices; and rejects all other users. If BS $j$ receives no more than $L_j$ bids, all these users will be put on BS $j$'s waiting list.

At the second stage, if a user is in a BS's waiting list, 
%in the previous round, 
it will keep on bidding the same BS with the same price to guarantee the highest utility. However, if a user gets rejected in the previous round, as being selfish, it will exclude the BS's that have rejected it and offers a price to its most desirable BS among the remaining ones. For the service provider, it adopts the same strategy. If the number of bids received for a BS outnumbers the load capacity of that BS, the service provider only keeps the $L_j$ most desirable users on the waiting list and rejects the others. It keeps all users on the waiting list if the number of offers received is less than a BS's load capacity. This two stages repeat until convergence is achieved.

\begin{lemma}
\label{lem10}
The sequence of bids made by a user is non-increasing in the user's preference list.
\end{lemma}
\begin{proof}
Before a user makes an offer, it computes the satisfaction of all the BS's to obtain a preference list.
Since a user aims to maximize its utility, it first proposes to the BS with the highest satisfaction. If it is rejected by the BS, it will propose to the BS with the second highest satisfaction, and so forth. Note that even if a user may be on the waiting list of a BS, it may be removed from that waiting list at a later stage. If that happens, this user will start bidding to other BS. 
A user will repeat this procedure until it is finally in a BS's serving list or rejected by all BS's. 
%So the sequence of bidding made by a user is non-increasing in the user's preference list.
This concludes the proof. 
\end{proof}

\begin{lemma}
\label{lem9}
The sequence of bids a BS put on the waiting list is non-decreasing in its preference list.
\end{lemma}
\begin{proof}
Given the fact any BS has a finite load capacity and $K \geq \sum_{j=1}^J L_j$, all the base station will have at least one user bidding to it at some stage of the game.
Since a BS aims to maximize its utility, it puts all the users who make an offer on the waiting list. On the condition that there are too many users, it will reject the users who it will never served. In the next round of game, the BS will often have more or at least the same amount of bids compared to its current waiting list. This means that the BS has more choices. The BS again only keeps the most profitable ones and reject or remove the others from the waiting list. So the sequence of bids a base station put on the list is non-decreasing in its preference list.
\end{proof}

\begin{theorem}
\label{thm4}
The repeated bidding game converges.
\end{theorem}
\begin{proof}
Based on~Lemmas~\ref{lem10} and~\ref{lem9}, we prove this theorem by contradiction. Suppose that this repeated game does converge. Then there must be a stage of the game that (i) there is a user $k$ and BS $j$ pair so that user $k$ is connected to another BS $j'$ or is not connected to any BS; (ii) user $k$ prefers BS $j$ to BS $j'$ or prefers 
%BS $j$ 
to be not connected; and (iii) BS $j$ prefers user $k$ to a user $k'$ who is on its serving list. 

Consider the case where user $k$ is served by BS $j'$. Since the sequence of bids made by a BS 
%put on the list 
is non-decreasing, 
%in its base station, 
it must be the case that user $k$ has never bidden to BS $j$ during the game. Otherwise, if user $k$ has bidden to BS $j$, BS $j$ would not have ended up with choosing $k'$ over $k$. 
%So user $k$ must have never bidden to BS $j$. 
In this case, user $k$ would never have bidden to BS $j'$ either, since user $k$ prefers $j$ to $j'$ and the bids (see Lemma~\ref{lem10}). However, user $k$ is now served by BS $j'$, user $k$ must have bidden to BS $j'$, which contradicts that user $k$ would never have bidden to BS $j'$.

The same reasoning holds for the case when user $k$ is not connected to any BS. If BS $j$ prefers $k$ to $k'$ on the serving list, BS $j$ would never reject user $k$ while keeping user $k'$.

%Therefore, the repeated game must converge.
Therefore, the game converges when every user is either on a waiting list or has been rejected by every BS, and the game will converge. 
\end{proof}

From the proof, we can actually see that the game terminates when the least popular BS becomes fully loaded.
\begin{theorem}
\label{thm3}
The outcome of the 
%above 
repeated bidding game is optimal for both the users and service provider. 
\end{theorem}
\begin{proof}
Suppose that the outcome of the game is not optimal for a user $k$, who is connected to BS $j$.
Then there must be another BS $j'$, which has higher ranking than BS $j$ in the preference list of user $k$ and has a serving list of users $\left\{ j'_1,j'_2,\cdots,j'_{L_{j'}} \right\}$. Since BS $j'$ serves these users, 
%$j'_1,j'_2,\cdots,j'_{L_{j'}}$, 
it means that BS $j'$ prefers them to user $k$ and BS $j'$ is at the top of the preference lists of these users. If at some stage, user $k$ is in the waiting list of BS $j$ (or it is inserted by force), 
%insert user $k$ by force to the waiting list of BS $j'$, 
the game must have not terminated. 

Since user $k$ is in the waiting list, then one of the final users $j'_1,j'_2,\cdots,j'_{L_{j'}}$ must be currently off the list, say user $j'_{L_{j'}}$. 
%is out of the list.
Then user $j'_{L_{j'}}$ will immediately bid for BS $j'$, since BS $j'$ is at the top of its preference list among all the remaining BS's. And BS $j'$ will remove user $k$ from its waiting list, since user $k$ has a lowest ranking in the preference list of BS $j'$. 
Thus when the repeated game terminates, the outcomes are optimal for each user. It is obvious that the outcome is also optimal for the service provider as well.
\end{proof}

\smallskip
%%%Another way to express Theorem \ref{thm3} is that users get the best they can get. However, this may not be true for the base stations. 
From Theorems~\ref{thm4} and~\ref{thm3}, we conclude that the game converges to the NE when the game terminates.

%%%\reminder{Consider bias to encourage offloading?}
%%%
%%%If $K = \sum^J_{j=1}{L_j}$, it is the \textit{Stable Marriage Matching} problem.
%%%If $K > \sum^J_{j=1}{L_j}$, it is the \textit{College Admission Matching} problem \cite{Gale1962}.

%-----------------------------------------------------------
\section{Simulation Study} \label{sec:sim}
%-----------------------------------------------------------

%In this section, we conduct numeric simulations to verify the efficacy of our proposed schemes. 
We validate the proposed user association schemes with simulations. 
Throughout the simulations, we assume $l_{j,k}=1/(1+(\frac{d_{j,k}}{40})^{3.5})$ for the path loss between a user and the massive MIMO BS, and $l_{j,k}=1/(1+(\frac{d_{j,k}}{40})^{4})$ for the path loss between a user and a small cell BS~\cite{Bethanabhotla}.
We assume that the power of small scale fading follows a uniform distribution from $\left[0.8,1\right]$. 
%For geometry of the BS's, we 
We fix the location of the massive MIMO BS at the center of the cell. The other BS's are randomly placed across in the cell. Users are randomly placed in the area. 
%dropped across the map. 
The other parameter settings are listed in Table~\ref{tab:sc}. The error bars in the plots are 95\% confidence intervals.

\begin{table} [!t]
\begin{center}
\caption{System Configuration}
\label{tab:sc}
\begin{tabular}{l|l||l|l}
\toprule
\textbf{Parameter} & \textbf{Value} & \textbf{Parameter} & \textbf{Value}\\
\midrule
$M_{massive}$ & $100$ & $M$ & $4$ \\
$L_{massive}$ & $10$ & $L$ & $4$ \\
$P_{massive}$ & $40$ dBm &P& $40$ dBm\\
Area & $1000\times 1000$ m$^2$ & $J$ & $11$\\
\bottomrule
\end{tabular}
\end{center}
\vspace{-0.1in}
\end{table}

%Fig. \ref{fig:centralu} 
Table~\ref{tab:centralu} presents a comparison of rate maximization with the optimal solution and the two proposed greedy algorithms. 
%Fig. \ref{fig:centralulog} 
Tabel~\ref{tab:centralulog} shows a comparison of rate maximization considering proportional fairness with the optimal solution and the two proposed greedy algorithms.
We can see from both tables
%figures 
that the optimal solution achieves the highest network utility. We also notice that as the number of users increases, the gaps between the optimal utility and the greedy solutions become more and more narrower. This is because that as there are more users, the user diversity effect becomes stronger. So the greedy algorithms and the optimal user association algorithm tend to produce similar solutions.

%​K = [50 100 150 200 250];
%Optimal Rate Maximization = [382.9266  483.4131  543.0618  572.2600  593.9714]
%Greedy Algorithm 1 =    [363.5394  480.0396  540.4285  571.0771  592.5843]
%Greedy Algorithm 2 =    [191.0061  279.5609  340.6012  371.4990  392.5354]

\begin{table} [!t]
\begin{center}
\caption{Rate Maximization of Centralized Control}
\label{tab:centralu}
\begin{tabular}{r|lllll}
\toprule
\textbf{K} & 50 & 100 & 150 & 200 & 250 \\
\midrule
Optimal Rate Maximation & 382.9 & 483.4 & 543.1 & 572.3 & 594.0 \\
Greedy Algorithm 1  & 363.5 & 480.0 & 540.4 & 571.1 & 592.6 \\
Greedy Algorithm 2  & 191.0 & 279.6 & 340.6 & 371.5 & 392.5 \\
\bottomrule
\end{tabular}
\end{center}
\vspace{-0.1in}
\end{table}

%Figure 3
%K = [50 100 150 200 250];
%Optimal Log Rate Maximization = [128.6505  155.5117  167.7080  172.6998  175.9847]
%Greedy Algorithm 1 = [115.0337  153.3264  166.8653  172.3864  175.6504]
%Greedy Algorithm 2 =            [67.5480   97.6558  121.9851  133.4823  139.6626]

\begin{table} [!t]
\begin{center}
\caption{Log Rate Utility of Centralized Control}
\label{tab:centralulog}
\begin{tabular}{r|lllll}
\toprule
\textbf{K} & 50 & 100 & 150 & 200 & 250 \\
\midrule
Optimal Log Rate Max. & 128.7 & 155.5 & 167.7 & 172.7 & 176.0 \\
Greedy Algorithm 1  & 115.0 & 153.3 & 166.9 & 172.4 & 175.7 \\
Greedy Algorithm 2  & 67.5  & 97.7  & 122.0 & 133.5 & 140.0 \\
\bottomrule
\end{tabular}
\end{center}
\vspace{-0.1in}
\end{table}

%-------------------------------------------------------
\begin{comment}

\begin{figure} [!t] %[thb]
\center{\includegraphics[width=3.2in, height=2.2in]{centraluti_ineq_bar.eps}}
\caption{Rate Maximization of Centralized Control.}
\label{fig:centralu}
%\vspace{-0.15in}
\end{figure}

\begin{figure} [!t] %[thb]
\center{\includegraphics[width=3.2in, height=2.2in]{centraluti_log_ineq_bar.eps}}
\caption{Log Rate Utility of Centralized Control.}
\label{fig:centralulog}
%\vspace{-0.15in}
\end{figure}

\end{comment}
%-------------------------------------------------------

\begin{figure} [!t] %[thb]
\center{\includegraphics[width=3.3in, height=2.3in]{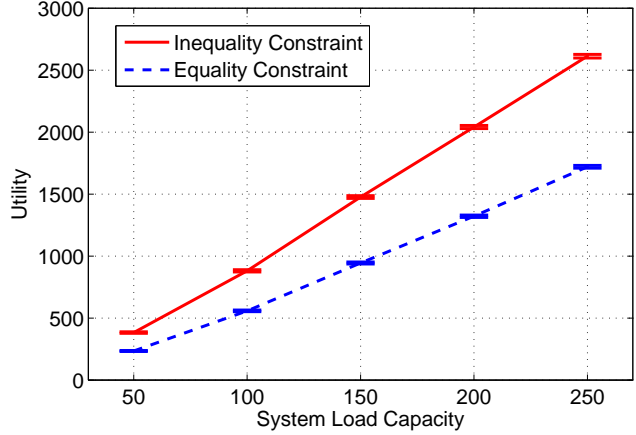}}
\caption{A comparison of the centralized algorithms with logarithmic rate utility under inequality and equality constraints.}
\label{fig:centralineqeqcmp}
%\vspace{-0.15in}
\end{figure}

Throughout this paper, the constraint for each user is $\sum_{j=1}^J x_{k_j} \leq 1$. It should provide upper bounds for the problem with the constraint $\sum_{j=1}^J x_{k_j}=1$.
A comparison of these two different constraints is presented in Fig.~\ref{fig:centralineqeqcmp}. For a fair comparison, we have exactly the same number of active users as the load capacity for all the BS's. For instance, when the system load capacity is $250$, we have $J=51$ BS's and $K=250$.
We can see that the inequality constraint problem indeed upper bounds the equality constraint problem. This is because the inequality constraint problem could eliminate the users whose rate is too low with a negative utility.

%Figure 5
%K = [50 100 150 200 250];
%Optimal Joint Resource Allocation and User Association=[42.2983   50.8491   55.7862   56.0959   61.9529]
%Greedy Algorithm 4 = [35.3293   37.9407   39.3203   39.8653   40.2636]
%Greedy Algorithm 5 = [35.0625   37.8518   39.3019   39.8222   40.2365]

\begin{table} [!t]
\begin{center}
\caption{Joint Resource Allocation and User Association}
\label{tab:jraua}
\begin{tabular}{r|lllll}
\toprule
\textbf{K} & 50 & 100 & 150 & 200 & 250 \\
\midrule
Optimal Joint Resource Allocation  & 42.3 & 50.8 & 55.8 & 56.1 & 62.0 \\
and User Association & & & & & \\
Greedy Algorithm 4  & 35.3 & 37.9 & 39.3 & 39.9 & 40.3 \\
Greedy Algorithm 5  & 35.1 & 37.9 & 39.3 & 39.8 & 40.2 \\
\bottomrule
\end{tabular}
\end{center}
\vspace{-0.1in}
\end{table}

%-------------------------------------------------------
\begin{comment}

\begin{figure} [!t] %[thb]
\center{\includegraphics[width=3.3in, height=2.3in]{joint_bar.eps}}
\caption{Joint resource allocation and user association.}
\label{fig:jraua}
%\vspace{-0.15in}
\end{figure}

\end{comment}
%-------------------------------------------------------

%Fig. \ref{fig:jraua} 
Table~\ref{tab:jraua} presents a comparison of the optimal joint resource allocation and user association algorithm and the two proposed greedy algorithms. We find that the optimal scheme achieves the highest utility. Moreover, the gap between the optimal scheme and the greedy schemes is quite large. We also consider the equality constraint problem as a benchmark for the comparison.
%, which is shown in Fig. \ref{fig:userasso1} and Fig. \ref{fig:userasso2}. 
For a fair comparison, we set the sum capacity of this system equal to the number of users. So there are totally $K=50$ active users in the system. 
%Fig. \ref{fig:userasso1} illustrates the optimal solution of problem (\ref{eqn:mmhn63eq}), with optimal network utility as $-59.8462$. Fig. \ref{fig:userasso2} illustrates the optimal solution of problem (\ref{eqn:mmhn63}), with optimal network utility as $29.5433$, which is much higher than $-59.8462$. We can also see that if we connect every user, some edge user with low rate will be harmful for the network utility.
The optimal solution of problem~{\bf P3-2}
%~(\ref{eqn:mmhn63eq}) 
achieves a network utility of $-59.8462$, while the optimal solution of problem~(\ref{eqn:mmhn63}) has a network utility of $29.5433$. We also found that if we connect every user, some edge users will be harmful for the network utility. 

%-----------------------------------------------------
\begin{comment}

\begin{figure} [!t] %[thb]
%\center{\includegraphics[width=\columnwidth]{userasso_equ.eps}}
\center{\includegraphics[width=3.2in, height=2.5in]{userasso_equ.eps}}
\caption{Optimal joint resource allocation and user association with equality constraint.}
\label{fig:userasso1}
%\vspace{-0.15in}
\end{figure}

\begin{figure} [!t] %[thb]  
%\center{\includegraphics[width=\columnwidth]{userasso_inequ.eps}}
\center{\includegraphics[width=3.2in, height=2.5in]{userasso_inequ.eps}}
\caption{Optimal joint resource allocation and user association with inequality constraint.}
\label{fig:userasso2}
%\vspace{-0.15in}
\end{figure}

\end{comment}
%-----------------------------------------------------

Fig.~\ref{fig:dist_bsp} shows the utility of the service provider and all users when the service provider sets the price (as in Section~\ref{subsec:spdtp}). It can be seen that the repeated game converges after $8$ rounds. Furthermore, the utility of all users is monotonically decreasing. That is because once a user's evaluation is known to the service provider, the service provider will set prices for the highest profit, which results in $0$ utility for that user. As discussed, the utility for the service provider is not monotonically increasing, since during the game, the load capacity constraint may be violated. Fig.~\ref{fig:dist_bsp_comp} plots the utilities of the service provider and users versus the number of users. We can see that as the number of user increases, utility of the service provider also increases. This is mainly due to the effect of multi-user diversity. We can also observe that the game terminates after about $8$ rounds no matter how many users are active.

\begin{figure} [!t] %[thb]  
\center{\includegraphics[width=3.3in, height=2.3in]{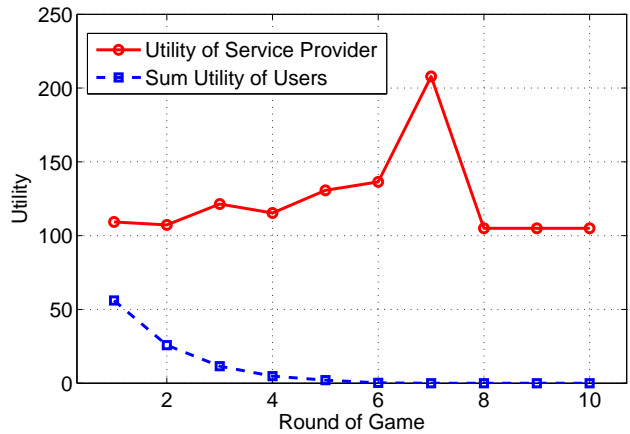}}
\caption{Convergence of the repeated game when the service provider sets the price and $K=100$.}
\label{fig:dist_bsp}
%\vspace{-0.15in}
\end{figure}

\begin{figure} [!t] %[thb]  
\center{\includegraphics[width=3.3in, height=2.3in]{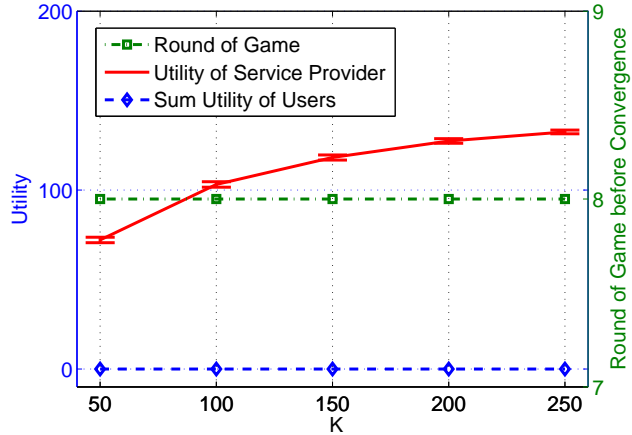}}
\caption{Utility of the service provider, utility of the users, and the number of rounds for convergence for systems with various numbers of users.}
\label{fig:dist_bsp_comp}
\vspace{-0.1in}
\end{figure}

Fig.~\ref{fig:conv} depicts the process of the game when users bid for BS's (as in Section~\ref{subsec:ub}).
Here we deploy $J=41$ BS's. The massive MIMO BS has $M=400$ 
%$M_{massive}=400$ 
antennas. There are $K=350$ users. The left-hand-side $y$-axis represents the load of the $41$ BS's. The right-hand-side $y$-axis represents the utility of the service provider.
We find the game converges in about $10$ rounds, and the utility of the service provider is monotonically increasing as the game continues.

\begin{figure} [!t] %[thb]  
\center{\includegraphics[width=3.3in, height=2.3in]{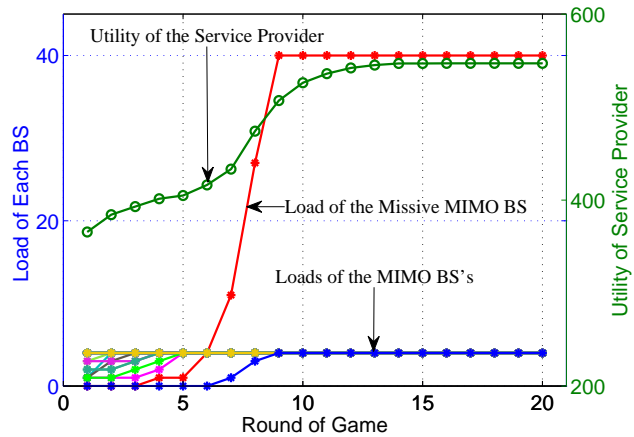}}
\caption{Convergence of the repeated game with respect to BS load when users bid.}
\label{fig:conv}
%\vspace{-0.15in}
\end{figure}

To encourage offloading from the macro BS, we consider rate bias for the BS's in this experiment. Specifically, we multiple the rate of the massive MIMO BS with a factor of $0.5$ to encourage connection to the PBS's. Fig.~\ref{fig:bsut} shows the result when configuration shown in 
%set the parameters according to 
Table~\ref{tab:sc}. It can be observed that the utility with rate bias is higher than the utility without considering rate bias. This result demonstrates the efficacy of rate bias and offloading.
%We could also see that the games terminate less than $8$ rounds of game.
It can be seen that both games terminate in less than $8$ rounds.

\begin{figure} [!t] %[thb]
\center{\includegraphics[width=3.3in, height=2.3in]{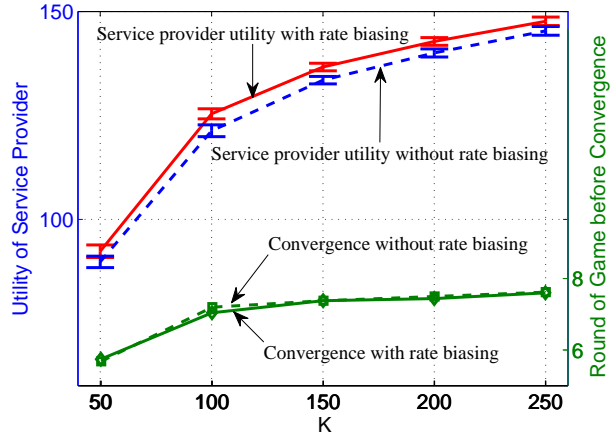}}
\caption{Utility of the service provider with or without rate bias, and convergence of the games under different numbers of users.}
\label{fig:bsut}
\vspace{-0.1in}
\end{figure}

%-----------------------------------------------------------
\section{Conclusions} \label{sec:con}
%-----------------------------------------------------------

In this paper, we investigated the user association problem in a massive MIMO HetNet from the centralized and distributed perspectives. Particularly, by leveraging totally unimodularity we developed optimal algorithms for rate maximization and rate maximization with proportional fairness problems. We also developed optimal algorithms to the joint resource allocation and user association problem with primal decomposition and dual decomposition. Modeling the behavior and interaction of the service provider and users with repeated games, we developed effective distributed algorithms with proven convergence to the NE. 
%proved that the games when the service provider sets the price or users bid for connection would converge to the Nash Equilibrium. 
%We have compared the proposed schemes to some heuristic schemes. 
Simulation results verify the efficacy of the proposed schemes.

\section*{Acknowledgment}

This work is supported in part by the US National Science Foundation (NSF) under Grants CNS-1247955 and CNS-1320664. Any opinions, findings, and conclusions or recommendations expressed in this material are those of the author(s) and do not necessarily reflect the views of the NSF.

\end{document}